\documentclass[aps,prx,twocolumn,notitlepage,superscriptaddress]{revtex4-2}

\usepackage{amsmath, amssymb, amsfonts}
\usepackage{amsthm}
\usepackage{mathtools}
\usepackage{mathrsfs}
\usepackage{graphicx}
\usepackage[caption=false]{subfig}
\usepackage{bm}
\usepackage{hyperref}
\usepackage{xcolor}

\DeclareMathOperator{\dist}{dist}

\newtheorem{theorem}{Theorem}

\newtheorem{lemma}{Lemma}

\theoremstyle{remark}
\newtheorem{remark}{Remark}

\hypersetup{
    colorlinks=true,
    linkcolor=blue,
    filecolor=magenta,
    urlcolor=cyan,
    pdftitle={The quantum harmonic oscillator from binary sequences},
    pdfauthor={Powers, Skrgic, Stojkovic}
}

\begin{document}

\title{The quantum harmonic oscillator from binary sequences}

\author{Samuel Powers}
\email{sampower@buffalo.edu}
\affiliation{Department of Physics, University at Buffalo, State University of New York, Buffalo, NY 14260, USA}
\author{Dino \v{S}krgi\'{c}}
\email{dinoskrg@buffalo.edu}
\affiliation{Department of Physics, University at Buffalo, State University of New York, Buffalo, NY 14260, USA}
\author{Dejan Stojkovic}
\email{ds77@buffalo.edu}
\affiliation{Department of Physics, University at Buffalo, State University of New York, Buffalo, NY 14260, USA}
\date{\today}

\begin{abstract}
The quantum harmonic oscillator (QHO) is normally built on a structure that contains a Hilbert space, ladder operators, and a Born rule. Here we derive it from counting. We adopt three information-theoretic postulates: information is carried by binary sequences of length $n$; only symbol counts, not the sequences themselves, are observable; and the transition measure must account for exactly the $2^n$-configuration capacity of the register. The first two postulates force the counting problem into the Hamming/Delsarte association scheme and endow it with an exact symmetry under relabeling of the two symbols ($0$ and $1$). Given the single modeling choice of this paper (an exchange-symmetric bilinear transition measure) we prove a rigidity theorem: relabeling symmetry and normalization together admit exactly one weighting of the hidden overlap between input and output sequences, namely the alternating sign $(-1)^t$, and every deviation from it strictly overshoots the capacity bound. The variable that carries the sign, i.e. the overlap between input and output sequences, or equivalently their Hamming distance $d_{ab}$, is a hidden quantum number of the transition as a whole, defined only by the two endpoints jointly and unknowable from either alone. Interference thus enters as forced alternating-sign bookkeeping over a nonlocal, relational variable $d_{ab}$, rather than as a separate dynamical ingredient. This forced weighting makes the transition probability proportional to the square of a Krawtchouk polynomial, whose difference equation has an exactly equally spaced spectrum and which converges to the QHO as $n\to\infty$. At fixed excitation, deviations from the QHO are $\mathcal{O}(1/n)$, and are testable in laboratory realizations of the oscillator, since any real system has finite information capacity. In this framework, interference, quadratic (Born-rule-like) probabilities, and the particle-hole symmetric oscillator spectrum are consequences of counting; only the bilinear form of the measure is assumed rather than derived.
\end{abstract}

\maketitle

\section{Introduction}
\label{sec:intro}

The quantum harmonic oscillator (QHO) is the most important exactly solvable model in physics, with direct and indirect applications throughout condensed matter physics~\cite{sakurai2017modernqm,walls2008quantum}, high energy physics~\cite{peskin1995introduction}, and quantum technologies~\cite{Ding:2015xyg,RevModPhys.75.281,RevModPhys.93.025005}. Its standard derivation, via creation and annihilation operators acting on a Hilbert space~\cite{sakurai2017modernqm}, takes the Hilbert space, the operator algebra, and the Born rule as given rather than derived. This raises a standing question in the foundations of quantum theory: to what extent can this structure be obtained from constraints that are not themselves quantum mechanical?

Two broad strategies have addressed this question. Axiomatic, top-down reconstructions~\cite{hardy2001quantumtheoryreasonableaxioms,dariano2017,chiribella2011informational} select quantum theory from a small set of operational or information-theoretic axioms (purification, local tomography, and related postulates), with interference and the Born rule entering as consequences of the axiom set as a whole rather than as separately motivated ingredients. Epistemically restricted, bottom-up hidden-variable models~\cite{spekkens2007toy,catani2017spekkens} instead posit an underlying ontic state space together with a restriction on what an observer can know about it, and show that a large fraction of the qualitative features of quantum mechanics follow~\cite{spekkens2007toy}. A third, largely independent tradition treats quantum dynamics itself as emergent from an underlying classical-statistical or cellular-automaton substrate evolving in time~\cite{wetterich2010qmfromclassical,thooft2016cellular,elze2014quantumness}.

The present work extends a fourth, more narrowly combinatorial route, developed in a companion series on binary-sequence models of quantum spin~\cite{Powers:2021rfg,Powers:2023ydg,Powers:2023sqf}. There, correlations between classical binary sequences, subject to an observability restriction on symbol counts, were shown to reproduce spin-$1/2$ quantum statistics, with the alternating sign responsible for interference traced to characters of a finite group. The natural question left open by that work is whether the same construction can reach the unbounded, equally spaced spectrum of the QHO, which requires an orthogonal polynomial family on the full binary Hamming scheme rather than on a single finite group orbit. We show that it does: the relevant family is the Krawtchouk polynomials~\cite{coleman2011krawtchouk,Feinsilver:2007oyk,kravchuk1929generalisations}, already well known as the finite-oscillator polynomials in the mathematical-physics literature~\cite{Atakishiyev2005,Atakishiyev_2001:JPA,Nikiforov_1991:inb}, whose three-term recurrence encodes the raising and lowering structure of the QHO in the large-$n$ limit.

We split the result into two parts, one forced and one assumed. Two information-theoretic postulates---that an observer has access only to symbol counts of a length-$n$ binary sequence, and that the total number of distinguishable configurations is exactly $2^n$---are enough by themselves to put the counting problem in the form of a Hamming/Delsarte association scheme~\cite{Delsarte1998AssociationSA}. This scheme has $n+1$ commuting relation matrices with a single joint eigenbasis, the characters of $(\mathbb{Z}_2)^n$, whose eigenvalues are the Krawtchouk polynomials, and none of this presupposes a probability rule, a Hilbert space, or a sign convention. To turn this structure into a transition probability we then adopt a bilinear measure built from a forward ensemble (input to output) and a backward ensemble (output to input). The motivation is the symmetry of the problem under exchanging the two endpoints of a transition and, more pointedly, the observability restriction itself: a forward ensemble alone fixes the input and ranges over outputs, tacitly treating the initial state as an exactly known sequence---knowledge that postulate (II) forbids, and the combinatorial analog of resolving a state beyond what the uncertainty principle permits---while the backward ensemble, which fixes the output and ranges over the admissible inputs, restores the same incomplete-knowledge status to both endpoints (Sec.~\ref{sec:rigidity}). Given that bilinear ansatz, we prove a rigidity theorem (Theorem~\ref{thm:rigidity} in Sec.~\ref{sec:rigidity}, with the full proof in Appendix~\ref{app:proof}): normalization and symbol-relabeling symmetry force the parity-alternating sign structure, and they do so constructively---any other choice of sign at any overlap value strictly overshoots the capacity bound, by an amount we compute exactly.

The result is as follows: given an information-capacity constraint on the Hamming scheme and a bilinear, exchange-symmetric transition ansatz, the sign structure responsible for interference is unique, and it is exactly the structure that produces the Krawtchouk finite oscillator and, in the continuum limit, the QHO. This explains why interference has this sign pattern and not some other one, \emph{given the ansatz}, rather than deriving the necessity of the ansatz itself from the postulates alone. A byproduct of the derivation is a sharp statement about \emph{where} interference resides: the forced sign is a function of the Hamming distance $d_{ab}$ between the input and the output (equivalently of their overlap $t$), a variable that behaves as a quantum number of the transition rather than of a state. It is nonlocal in the precise sense that its value requires joint knowledge of the initial and final states and cannot be reconstructed from either endpoint alone, and by postulate (II) it is hidden from the observer. Interference then appears as mandatory bookkeeping over this relational hidden variable, a reading we develop in Secs.~\ref{sec:krawtchouk} and~\ref{sec:discussion}.

We adopt three postulates:

\begin{itemize}
\item[\bf (I)] Information about the system of interest is encoded in binary sequences of length $n$.
\item[\bf (II)] Individual sequences are not observable. An observer has access only to symbol counts---the counts of a single sequence, and the joint counts of a pair of sequences---together with the cardinalities of the sets of sequences consistent with those counts. A transition probability between an input and an output can therefore depend on the pair only through such counts.
\item[\bf (III)] The transition measure assigned to a fixed input, summed over all outputs, must exhaust the information capacity of the register exactly: it accounts for exactly $2^n$ distinguishable configurations, neither more (which would overcount observationally indistinguishable configurations) nor fewer (which would leave configurations unaccounted for).
\end{itemize}

Postulates (I) and (II) are purely combinatorial statements about what information an observer can access; they do not by themselves specify a probability rule. Postulate (III) is normalization, expressed in capacity form: dividing the total measure by $2^n$, it reads $\sum_{\rm outputs} P = 1$. The first two postulates force the association scheme structure and an exact bit-relabeling symmetry of the counting problem; the third fixes the overall scale. Together, as we prove below, they leave exactly one admissible sign structure.

The paper is organized as follows. Section~\ref{sec:scheme} sets up the counting problem, exhibits the exact failure of unsigned counting, and identifies the structure that postulates (I) and (II) force with no further input: the Hamming association scheme and its relabeling symmetry. Section~\ref{sec:rigidity} introduces the bilinear ansatz, states the rigidity theorem, and sketches its proof; the complete proof is in Appendix~\ref{app:proof}. Section~\ref{sec:krawtchouk} evaluates the forced measure in closed form, identifies it with the Krawtchouk finite oscillator, and takes the continuum limit to the QHO. Section~\ref{sec:numerics} presents numerical checks, including an entropy feature at the particle-hole self-dual point for which we give an exact explanation. Appendix~\ref{app:n4} works the $n=4$ case by direct integer counting.

\section{Model}
\label{sec:scheme}

Let $S^1(n)$ denote the set of all binary sequences of length $n \in \mathbb{N}_0$. Tilde notation denotes symbol counts: a sequence $s^1 \in S^1(n,\tilde{1})$ contains $\tilde{1} \in \mathbb{N}_0$ ones and $\tilde{0} \equiv n - \tilde{1}$ zeros, so that
\begin{align}
\label{s1}
    |S^1(n,\tilde{1})| = \frac{n!}{\tilde{1}!\,\tilde{0}!} = \binom{n}{\tilde 1}, \qquad |S^1(n)| = 2^n.
\end{align}

We model a measurement as a transition $s^1_a \to s^1_b$ (Alice to Bob) between an input sequence $s^1_a \in S^1(n,\tilde{1}_a)$ and an output sequence $s^1_b \in S^1(n,\tilde{1}_b)$. The object of interest is the probability distribution over output counts $\tilde{1}_b$, conditioned on fixed input counts $\tilde{1}_a$. Where it lightens the notation, we abbreviate $a\equiv\tilde 1_a$ and $b\equiv\tilde 1_b$.

The element-wise pairing $s^1_a \otimes s^1_b$ is a base-4 sequence over $\{00,01,10,11\}$ with symbol frequencies $(f_{00},f_{01},f_{10},f_{11})$. These four frequencies are completely determined by $n$, $\tilde{1}_a$, $\tilde{1}_b$, and one additional variable, which may be taken either as the Hamming distance $d_{ab} \equiv f_{01}+f_{10}$ or, equivalently, as the overlap
\begin{align}
\label{eq:overlap}
    t \equiv f_{11} = |s_a \cap s_b|,
\end{align}
the number of positions at which both sequences carry a $1$. The two parametrizations are related within each sector by the fixed affine bijection $d_{ab} = \tilde{1}_a + \tilde{1}_b - 2t$; Table~\ref{tab:frequencies} lists all four frequencies in terms of $t$. We use $t$ throughout, because it is defined symmetrically for any pair of sequences and is manifestly invariant under $s_a \leftrightarrow s_b$ exchange; the choice is a matter of convenience only, as we show explicitly in Sec.~\ref{sec:rigidity}.

\begin{table}[ht!]
\centering
\begin{ruledtabular}
\begin{tabular}{c@{\hspace{1em}}c@{\hspace{1em}}c}
Symbol & Count (via $t$) & Count (via $d_{ab}$) \\
\hline
$f_{11}$ & $t$ & $(\tilde{1}_a + \tilde{1}_b - d_{ab})/2$ \\
$f_{10}$ & $\tilde{1}_a - t$ & $(\tilde{1}_a - \tilde{1}_b + d_{ab})/2$ \\
$f_{01}$ & $\tilde{1}_b - t$ & $(\tilde{1}_b - \tilde{1}_a + d_{ab})/2$ \\
$f_{00}$ & $n - \tilde{1}_a - \tilde{1}_b + t$ & $n - (\tilde{1}_a + \tilde{1}_b + d_{ab})/2$ \\
\end{tabular}
\end{ruledtabular}
\caption{Base-4 symbol frequencies of the pairing $s^1_a\otimes s^1_b$, expressed through the overlap $t=f_{11}$ and, equivalently, through the Hamming distance $d_{ab}=\tilde 1_a+\tilde 1_b-2t$.}
\label{tab:frequencies}
\end{table}

The cardinality of the measurement set $S^2(n,\tilde{1}_a,\tilde{1}_b,t)$ of all base-4 sequences consistent with these counts is
\begin{align}
\label{S2 in terms of f00 etc}
    |S^2(n,\tilde{1}_a,\tilde{1}_b,t)| = \frac{n!}{f_{00}!\,f_{01}!\,f_{10}!\,f_{11}!}.
\end{align}
Requiring all four frequencies in Table~\ref{tab:frequencies} to be non-negative integers restricts the overlap to the admissible range
\begin{align}
\label{eq:Lambda}
    \max(0,\,\tilde 1_a+\tilde 1_b-n)\;\le\; t \;\le\; \min(\tilde 1_a,\tilde 1_b),
\end{align}
and we write $\Lambda(n,\tilde 1_a,\tilde 1_b)\subset\mathbb{Z}$ for this set.
With the convention $\binom{m}{j}=0$ outside $0\le j\le m$, sums over $t$ may equivalently be taken over all integers.

Fixing $n$ and $\tilde{1}_a$, we treat $\tilde{1}_b$ as the random variable and $t$ as a hidden (nuisance) parameter to be summed over, using both a forward ensemble ($s^1_a \to s^1_b$) and a backward ensemble ($s^1_b \to s^1_a$), obtained by exchanging the roles of Alice and Bob. The reason for carrying the backward ensemble alongside the forward one is epistemic and is spelled out below Eq.~\eqref{eq:probability}: postulate (II) denies sequence-level knowledge of \emph{either} endpoint, and only the pair of ensembles treats input and output on the same incomplete-knowledge footing. The elementary ensemble cardinalities, normalizing $|S^2|$ by the size of the appropriate conditioning set, are
\begin{align}
\label{eq abs epsilon}
    |\varepsilon_a| \equiv \frac{|S^2|}{|S^1(n,\tilde{1}_a)|}, \qquad |\varepsilon_b| \equiv \frac{|S^2|}{|S^1(n,\tilde{1}_b)|}.
\end{align}
Substituting the frequencies of Table~\ref{tab:frequencies} into Eq.~\eqref{S2 in terms of f00 etc} and simplifying the factorials gives the closed forms
\begin{align}
\label{eq:eps closed}
    |\varepsilon_a| = \binom{\tilde 1_a}{t}\binom{n-\tilde 1_a}{\tilde 1_b - t},
    \qquad
    |\varepsilon_b| = \binom{\tilde 1_b}{t}\binom{n-\tilde 1_b}{\tilde 1_a - t}.
\end{align}
The interpretation of $|\varepsilon_a|$ is elementary: given a fixed input sequence with $\tilde 1_a$ ones, it counts the number of output sequences with $\tilde 1_b$ ones that overlap the input in exactly $t$ positions---choose which $t$ of the input's ones are shared, and which $\tilde 1_b - t$ of the input's zeros are flipped.

\subsection{Unsigned counting (no interference) fails}
\label{sec:naive}

A natural first attempt at a transition probability is to weight every hidden channel $t$ equally: sum each ensemble over $t\in\Lambda$ with unit weight and take the normalized product,
\begin{align}
\label{eq:naive}
    P^{\rm cl}(\tilde 1_b\mid n,\tilde 1_a)
    \equiv \frac{1}{2^n}
    \Big[\sum_{t\in\Lambda} |\varepsilon_a|\Big]
    \Big[\sum_{t\in\Lambda} |\varepsilon_b|\Big].
\end{align}
This fails, and it fails by an exactly computable amount. By the Vandermonde convolution, each unsigned sum evaluates in closed form,
\begin{align}
\label{eq:vandermonde}
\begin{split}
    \sum_{t\in\Lambda}\binom{\tilde 1_a}{t}\binom{n-\tilde 1_a}{\tilde 1_b-t} &= \binom{n}{\tilde 1_b},\\
    \sum_{t\in\Lambda}\binom{\tilde 1_b}{t}\binom{n-\tilde 1_b}{\tilde 1_a-t} &= \binom{n}{\tilde 1_a},
\end{split}
\end{align}
so that $P^{\rm cl}(\tilde 1_b\mid n,\tilde 1_a) = \binom{n}{\tilde 1_b}\binom{n}{\tilde 1_a}/2^n$ and
\begin{align}
\label{eq:overshoot}
    \sum_{\tilde 1_b=0}^{n} P^{\rm cl}(\tilde 1_b\mid n,\tilde 1_a) = \binom{n}{\tilde 1_a}.
\end{align}
The unsigned measure therefore exceeds the capacity bound of postulate (III) by exactly the factor $\binom{n}{\tilde 1_a}$, which is $>1$ for every $0<\tilde 1_a<n$ and grows polynomially, as $n^{\tilde 1_a}/\tilde 1_a!$, at fixed excitation. The physical reading is that $t$ (equivalently $d_{ab}$) is a nonlocal quantity---nonlocal in the precise sense that its value is a joint property of the input \emph{and} the output sequence, unknowable from either endpoint alone---and it is hidden from an observer restricted by postulate (II): all configurations sharing $(\tilde 1_a,\tilde 1_b)$ but differing in $t$ are observationally indistinguishable, and unit-weight summation multiply counts them. Figure~\ref{fig:capacity}(a) displays this overshoot. Some non-trivial weighting $w(t)$ is therefore unavoidable if the resulting object is to be a valid probability at all; the question is which weighting, and whether it is unique. Sections~\ref{sec:forced} and~\ref{sec:rigidity} answer both.

\begin{figure}[t]
    \centering
    % NOTE: rename image files before submission (e.g. fig_capacity_overshoot.png);
    % paths kept from the previous draft so the document compiles as-is.
    \includegraphics[width=0.95\columnwidth]{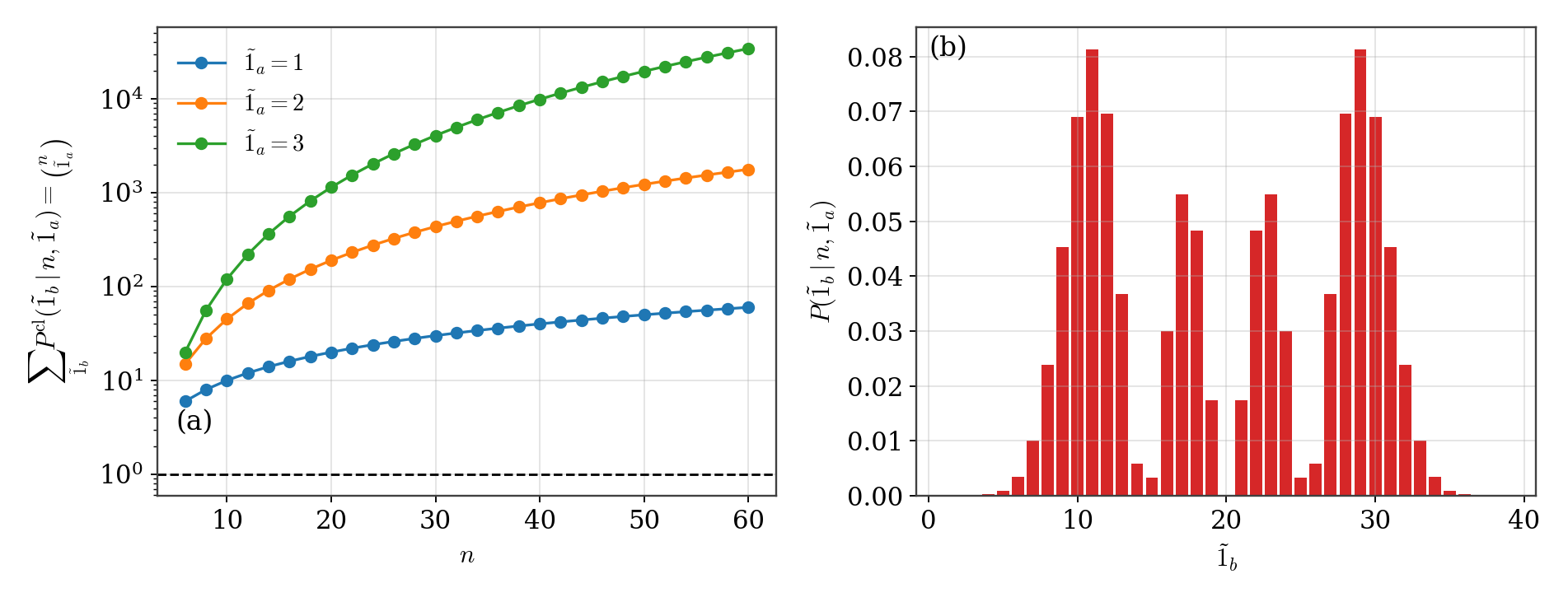}
    \caption{(a) The unsigned sum $\sum_{\tilde{1}_b} P^{\rm cl}$ of Eq.~\eqref{eq:naive}, which evaluates exactly to $\binom{n}{\tilde 1_a}$ [Eq.~\eqref{eq:overshoot}], exceeds the capacity bound of $1$ for every $0<\tilde 1_a<n$ and grows with $n$ at fixed excitation level (shown for $\tilde 1_a=1,2,3$); a non-trivial weighting is therefore unavoidable, not optional. (b) The signed weighting forced by Theorem~\ref{thm:rigidity} (here for $n=40$, $\tilde{1}_a=3$) produces a valid, normalized probability distribution with the correct $k=3$ nodal structure.}
    \label{fig:capacity}
\end{figure}

\subsection{The postulates force the Hamming scheme and its relabeling symmetry}
\label{sec:forced}

We now make precise what, in the postulates above, is genuinely forced and requires no further assumption. This part is deliberately restricted to statements that do not reference probability at all: it is pure counting and representation theory on $\{0,1\}^n$.

The direct connection to the Delsarte association scheme~\cite{Delsarte1998AssociationSA} is made explicit by identifying $S^1(n)$ with the elementary abelian group $(\mathbb{Z}_2)^n$. For $d=0,\ldots,n$, let $R_d \subset S^1(n)\times S^1(n)$ be the set of pairs at Hamming distance $d$, and let $A_d$ be the associated $2^n\times 2^n$ relation (adjacency) matrix, $(A_d)_{s^1_a,s^1_b} = 1$ if $(s^1_a,s^1_b)\in R_d$ and $0$ otherwise. The collection $\{R_d\}_{d=0}^n$ is the Hamming association scheme.

That $\{R_d\}_{d=0}^n$ forms an association scheme, and in particular that the $A_d$ pairwise commute, follows purely from the group structure of $(\mathbb{Z}_2)^n$ and the fact that Hamming distance is a class function of the group difference $s^1_a \oplus s^1_b$; it requires nothing beyond postulate (I) (sequences of length $n$) and postulate (II) (only counts, i.e. only which distance class a pair falls into, are observable to begin with). This is the sense in which the scheme itself is forced: postulates (I) and (II) are precisely the statement that the observable data about a pair $(s^1_a,s^1_b)$ is which $R_d$ it belongs to, nothing finer.

The commuting family $\{A_d\}_{d=0}^n$ is simultaneously diagonalized by the characters of $(\mathbb{Z}_2)^n$,
\begin{align}
    \chi_S(x) = (-1)^{\sum_{i\in S} x_i}, \qquad S \subseteq \{1,\ldots,n\},
\end{align}
with eigenvalue
\begin{align}
\label{eq:scheme eigenvalue}
    A_d\,\chi_S = K_d(|S|;n)\,\chi_S,
\end{align}
where $K_d(x;n) \equiv \sum_{j=0}^d (-1)^j\binom{x}{j}\binom{n-x}{d-j}$ is the Krawtchouk polynomial of degree $d$. This is a standard fact about the Bose--Mesner algebra of the Hamming scheme~\cite{Delsarte1998AssociationSA}: the $\chi_S$ are exactly the irreducible characters of the group $(\mathbb{Z}_2)^n$, every $A_d$ is a class-sum operator, and characters of an abelian group simultaneously diagonalize every element of its group algebra. The eigenvalue $K_d(|S|;n)$ depends on $S$ only through $|S|$, so we write it as a function of $|S|$ and $n$ alone from here on.

An important structural fact follows immediately: the eigenvalues $K_d(\,\cdot\,;n)$ are real, because $(\mathbb{Z}_2)^n$ is an elementary abelian $2$-group---every element is its own inverse, so every irreducible character is $\pm1$-valued---and for $d\geq 1$ they change sign as functions of their argument. Figure~\ref{fig:scheme}(a) makes this checkerboard sign pattern explicit for $n=32$. The sign structure that will become interference is therefore already latent in the joint eigenbasis that postulates (I) and (II) force; what remains to be shown is that a valid transition probability has no choice but to use it.

\begin{figure}[t]
    \centering
    \includegraphics[width=0.95\columnwidth]{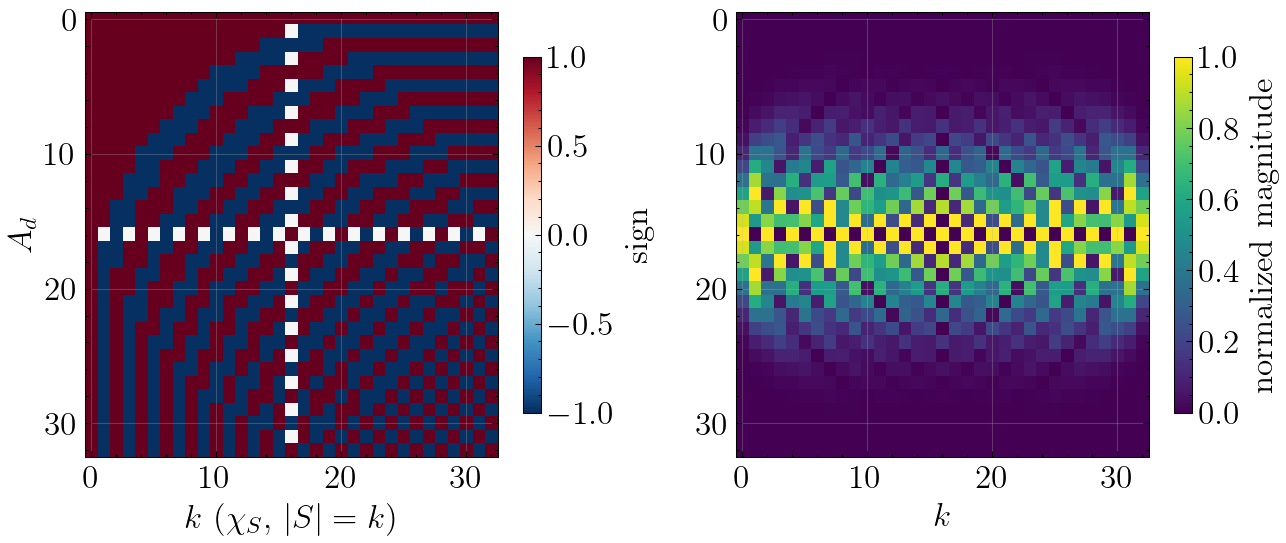}
    \caption{(a) Sign of $K_d(k;n)$ for $n=32$: the $\pm1$ character structure of the joint eigenbasis forced by postulates (I) and (II). (b) Column-normalized magnitude $|K_d(k;n)|$, showing the fringe pattern whose zero-crossings in (a) become the nodal structure of the Krawtchouk/QHO eigenstates derived below.}
    \label{fig:scheme}
\end{figure}

Postulates (I) and (II) reference only counts of symbols; nothing in them distinguishes the symbol $1$ from the symbol $0$. Consequently, the complementation map $C$ that flips every bit of both $s_a$ and $s_b$ is an exact symmetry of the counting problem: it preserves the Hamming distance, $\dist(s_a^c, s_b^c) = \dist(s_a, s_b)$, permutes the four symbol frequencies as $f_{00}\leftrightarrow f_{11}$, $f_{01}\leftrightarrow f_{10}$, and leaves $|S^2|$ unchanged, while sending the sector labels $(\tilde 1_a,\tilde 1_b) \to (n-\tilde 1_a,\,n-\tilde 1_b)$. Any transition probability built from this data must therefore satisfy the bit-relabeling invariance
\begin{align}
\label{eq:relabeling}
    P(\tilde{1}_b \mid n, \tilde{1}_a) = P(n-\tilde{1}_b \mid n, n-\tilde{1}_a),
\end{align}
on the same footing as the Alice/Bob exchange symmetry that motivates the bilinear ansatz of the next section. Note also, for later use, that complementation acts on the hidden overlap as $t = f_{11} \mapsto f_{00} = t + (n - \tilde 1_a - \tilde 1_b)$: relabeling the symbols \emph{shifts} the overlap by the sector-dependent amount $n-\tilde 1_a-\tilde 1_b$. This shift action is the mechanism through which Eq.~\eqref{eq:relabeling} constrains the $t$-dependence of any admissible weighting.

\subsection{The rigidity theorem}
\label{sec:rigidity}

We now introduce the one genuinely modeling ingredient of this paper and state precisely what follows once it is adopted. Motivated by the physical symmetry between the two endpoints of a transition---exchanging the roles of Alice and Bob should not be privileged in either direction---we construct the transition measure as a product of a forward and a backward weighted cardinality,
\begin{align}
\label{eq:bilinear-ansatz}
\begin{split}
    \Upsilon_a(n,\tilde{1}_a,\tilde{1}_b) &\equiv \sum_{t\in\Lambda} w(t)\,|\varepsilon_a|,\\
    \Upsilon_b(n,\tilde{1}_a,\tilde{1}_b) &\equiv \sum_{t\in\Lambda} w(t)\,|\varepsilon_b|,
\end{split}
\end{align}
with the candidate probability defined as
\begin{align}
    \label{eq:probability}
    P(\tilde{1}_b\mid n,\tilde{1}_a) \equiv \frac{\Upsilon_a\,\Upsilon_b}{2^n},
\end{align}
for some real weighting function $w(t)$ to be determined.

The bilinear structure can be motivated directly from the observability restriction. Consider building the measure from the forward ensemble alone: fix an input sequence and count the admissible outputs. Such a construction tacitly treats the initial state as known exactly---as a single, fully resolved sequence---while all uncertainty is assigned to the final state. But postulate (II) denies any observer this knowledge: the input, no less than the output, is accessible only through its symbol counts, and pretending otherwise would privilege information that no measurement in the model can supply. In quantum-mechanical language, it would be the combinatorial analog of specifying the initial state with a precision the uncertainty principle forbids. The division by $|S^1(n,\tilde 1_a)|$ in Eq.~\eqref{eq abs epsilon} already softens this for the forward ensemble---the input enters as a uniform ensemble over all sequences consistent with $\tilde 1_a$, not as an individually known sequence---but a residual asymmetry remains: the forward ensemble marginalizes the input and counts the outputs, while nothing yet constrains knowledge of the final state in the same way. The backward ensemble compensates: it fixes the output count and ranges over the inputs consistent with it. In the product of Eq.~\eqref{eq:probability} the two endpoints of the transition then enter on exactly the same incomplete-knowledge footing---each specified by its observable counts and nothing more---so the epistemic restriction of postulate (II) is applied symmetrically to preparation and outcome. This is also the sense in which the exchange symmetry invoked above is not merely aesthetic: privileging either direction would amount to claiming sequence-level knowledge of one endpoint.

Equation~\eqref{eq:bilinear-ansatz} is the simplest bilinear form with this two-sided epistemic symmetry that is compatible with postulates (I)--(III), but we have not shown, and do not claim, that it is the only construction compatible with those postulates. What we do show is that once Eqs.~\eqref{eq:bilinear-ansatz}--\eqref{eq:probability} are adopted, the weighting $w$ is no longer a free choice at all.

Before stating the theorem, we dispose of an apparent ambiguity in the choice of summation variable. Since $\tilde{1}_a=f_{11}+f_{10}$ and $\tilde{1}_b=f_{11}+f_{01}$ are fixed within a sector, the parities of the four symbol frequencies are locked together: $(-1)^{f_{10}}=(-1)^{\tilde{1}_a}(-1)^{f_{11}}$, $(-1)^{f_{01}}=(-1)^{\tilde{1}_b}(-1)^{f_{11}}$, and $(-1)^{f_{00}}=(-1)^{n-\tilde{1}_a-\tilde{1}_b}(-1)^{f_{11}}$. A parity weighting hung on any one of the four frequencies therefore differs from one hung on $f_{11}$ only by a global sign depending on the fixed sector labels $(\tilde{1}_a,\tilde{1}_b)$, not on the summation variable itself. Because that prefactor multiplies $\Upsilon_a$ and $\Upsilon_b$ identically, it squares to unity in the product Eq.~\eqref{eq:probability}, leaving $P$ invariant across all four choices. This is a gauge-like freedom of the parametrization; we fix the gauge by weighting $t=f_{11}$, which is manifestly symmetric under $s_a\leftrightarrow s_b$.

\begin{theorem}[Rigidity of the weighting]
\label{thm:rigidity}
Let $n\ge 2$, let $w:\{0,\ldots,n\}\to\mathbb{R}$, and let $\Upsilon_a$, $\Upsilon_b$, and $P$ be defined by Eqs.~\eqref{eq:bilinear-ansatz}--\eqref{eq:probability}. Suppose that
\begin{itemize}
\item[\bf (N)] $P$ is normalized in every sector: $\sum_{\tilde 1_b=0}^n P(\tilde 1_b\mid n,\tilde 1_a) = 1$ for every $\tilde 1_a \in \{0,\ldots,n\}$; and
\item[\bf (C)] $P$ is invariant under bit relabeling, Eq.~\eqref{eq:relabeling}.
\end{itemize}
Then
\begin{align}
\label{eq:forced w}
w(t) = \sigma\,(-1)^t \quad \text{for all } 0\le t\le n-1,
\end{align}
for a single overall sign $\sigma\in\{\pm1\}$, and $w(n)=\pm(-1)^n$. Neither $\sigma$ nor the sign of $w(n)$ affects $P$: the overall sign cancels in the bilinear product, and $w(n)$ enters $P$ only quadratically, through the single sector $(\tilde 1_a,\tilde 1_b)=(n,n)$. Consequently $P$ itself is unique, non-negativity holding automatically for the forced weighting.
\end{theorem}

The proof is elementary but instructive, and we give it in full in Appendix~\ref{app:proof}. Its three steps are worth sketching here, because each step uses a different hypothesis and the division of labor is the conceptual content of the theorem.

\emph{Step 1 (normalization pins the scale).} In the ground sector $\tilde 1_a=0$ only the $t=0$ channel exists, and normalization forces $w(0)^2=1$; we set $w(0)=+1$ without loss of generality.

\emph{Step 2 (relabeling symmetry pins the magnitude).} Applying the relabeling invariance (C) to the sector pair $(\tilde 1_a,\tilde 1_b)=(0,b) \leftrightarrow (n,n-b)$ yields $P(n-b\mid n,n) = w(n-b)^2\binom{n}{b}/2^n$, which must equal $P(b\mid n,0)=\binom{n}{b}/2^n$. Hence $w(t)^2=1$ for \emph{every} $t$: the weighting is a pure sign pattern, $w(t)\in\{\pm1\}$. This step is where relabeling symmetry does irreplaceable work; normalization alone is not enough. Indeed, at $\tilde 1_a=1$ normalization admits, besides $w(1)=-1$, the spurious real root $w(1)=-(n-3)/(n+1)$ [Appendix~\ref{app:proof}, Remark~\ref{rem:spurious}], which has $|w(1)|\neq1$ and is eliminated precisely by (C).

\emph{Step 3 (normalization pins every sign, with a strict overshoot).} Proceeding by induction on the sector $\tilde 1_a=1,2,\ldots,n-1$, suppose $w(t)=(-1)^t$ for all $t<\tilde 1_a$. Both $\Upsilon_a$ and $\Upsilon_b$ in sector $\tilde 1_a$ involve $w(t)$ only for $t\le\tilde 1_a$, so a single sign is undetermined. A closed-form computation (Appendix~\ref{app:proof}) shows that choosing the wrong sign, $w(\tilde 1_a)=-(-1)^{\tilde 1_a}$, gives
\begin{align}
\label{eq:overshoot-formula}
\begin{split}
    \sum_{\tilde 1_b} \Upsilon_a\Upsilon_b &= 2^n + 4\big(S_3 - 2^{\,n-\tilde 1_a}\big),\\
    S_3 &= \sum_{j\ge0}\binom{\tilde 1_a}{j}\binom{n-\tilde 1_a}{j}\,2^{\,n-\tilde 1_a-j},
\end{split}
\end{align}
and $S_3 > 2^{n-\tilde 1_a}$ strictly for $1\le\tilde 1_a\le n-1$. Every deviation from the alternating sign therefore strictly \emph{overshoots} the capacity bound of postulate (III)---the same failure mode, in miniature, as the fully unsigned measure of Eq.~\eqref{eq:overshoot}. The alternating sign is thus not merely one solution; it is the unique survivor, and the obstruction to every alternative is quantitative.

\begin{remark}
The theorem's hypotheses are exactly the content of postulates (III) and (II) respectively, evaluated on the ansatz: (N) is the capacity bound in normalized form, and (C) is the statement that nothing distinguishes the two symbols. The proof uses only the single slice $\tilde 1_a\in\{0,n\}$ of (C); the forced $P$ then satisfies Eq.~\eqref{eq:relabeling} in full, as one checks directly from the reflection identities of the Krawtchouk polynomials (Appendix~\ref{app:proof}). The shift action $t\mapsto t+(n-\tilde 1_a-\tilde 1_b)$ of complementation on the overlap, noted below Eq.~\eqref{eq:relabeling}, is what ties the value of $w$ at large $t$ to its value at small $t$ and makes the magnitude constraint of Step 2 possible.
\end{remark}

\begin{remark}
The theorem shows that, given the bilinear ansatz of Eqs.~\eqref{eq:bilinear-ansatz}--\eqref{eq:probability}, the weighting responsible for the interference-like sign structure is the only one consistent with the relabeling symmetry forced by postulates (I)--(II) and the capacity bound of postulate (III). It does not show that the ansatz itself is the only way to build a valid transition probability out of the $|S^2|$ cardinalities; that would require showing that every admissible construction factors through the Bose--Mesner algebra of the scheme, which we do not attempt here.
\end{remark}

\subsection{The Krawtchouk transition measure and its continuum limit}
\label{sec:krawtchouk}

We now evaluate the forced measure in closed form and show that it is the Krawtchouk finite-oscillator distribution, converging to the QHO in the continuum limit.

The Krawtchouk polynomial of degree $k$ evaluated at integer $x$ with parameter $n$ is~\cite{coleman2011krawtchouk,kravchuk1929generalisations}
\begin{align}
\label{krawtchouk polynomial definition}
    K_k(x;n) \equiv \sum_{j=0}^k (-1)^j \binom{x}{j}\binom{n-x}{k-j}.
\end{align}
These polynomials are orthogonal on $\{0,\ldots,n\}$ with respect to the binomial weight $\binom{n}{x}2^{-n}$,
\begin{align}
\label{orthogonality}
    \sum_{x=0}^n \binom{n}{x} K_k(x;n)K_l(x;n) = 2^n\binom{n}{k}\delta_{kl},
\end{align}
satisfy the duality relation
\begin{align}
\label{eq:duality}
\binom{n}{x}K_k(x;n) = \binom{n}{k}K_x(k;n),
\end{align}
the reflection identities
\begin{align}
\label{eq:reflections}
\begin{split}
K_k(n-x;n) &= (-1)^k\, K_k(x;n),\\
K_{n-k}(x;n) &= (-1)^x\, K_k(x;n),
\end{split}
\end{align}
the three-term recurrence
\begin{align}
\label{3 term recurrence}
    (k+1)K_{k+1} = (n-2x)K_k - (n-k+1)K_{k-1},
\end{align}
and the difference equation
\begin{align}
\label{difference equation}
    (n-x)K_k(x{+}1;n) + xK_k(x{-}1;n) = (n-2k)K_k(x;n).
\end{align}
These relations follow from general properties of binary Krawtchouk polynomials~\cite{coleman2011krawtchouk}; the reflection identities in Eq.~\eqref{eq:reflections}, which we use below to explain both particle-hole symmetry and an entropy feature at the self-dual point, follow in one line from the generating function $\sum_k K_k(x;n)z^k=(1-z)^x(1+z)^{n-x}$.

With $w(t)=(-1)^{t}$ as forced by Theorem~\ref{thm:rigidity}, the weighted ensembles evaluate to single Krawtchouk polynomials. Indeed, inserting the closed form Eq.~\eqref{eq:eps closed} into the ansatz,
\begin{align}
    \Upsilon_a = \sum_{t\in\Lambda} (-1)^t \binom{\tilde{1}_a}{t}\binom{n-\tilde{1}_a}{\tilde{1}_b-t} = K_{\tilde{1}_b}(\tilde{1}_a;n),
\end{align}
which is the definition~\eqref{krawtchouk polynomial definition} term by term; the companion identity $\Upsilon_b = K_{\tilde{1}_a}(\tilde{1}_b;n)$ follows by exchanging $\tilde 1_a\leftrightarrow\tilde 1_b$. The measure of Eq.~\eqref{eq:probability} is therefore
\begin{align}
\label{eq:P symmetric form}
P(\tilde 1_b\mid n,\tilde 1_a) = \frac{K_{\tilde 1_b}(\tilde 1_a;n)\,K_{\tilde 1_a}(\tilde 1_b;n)}{2^n},
\end{align}
and applying the duality relation Eq.~\eqref{eq:duality} to the first factor,
\begin{align}
    \label{P in terms of Krawtchouk}
    P(\tilde{1}_b\mid n,\tilde{1}_a) = \frac{\binom{n}{\tilde{1}_b}}{\binom{n}{\tilde{1}_a}}\,\frac{\big[K_{\tilde{1}_a}(\tilde{1}_b;n)\big]^2}{2^n},
\end{align}
proportional to the square of a single Krawtchouk polynomial. Positivity is now manifest; normalization follows from the orthogonality relation~\eqref{orthogonality} with $k=l=\tilde{1}_a$. Three structural properties are worth recording explicitly. First, Eq.~\eqref{eq:P symmetric form} is symmetric under $\tilde 1_a\leftrightarrow\tilde 1_b$: the probability of the transition does not depend on which endpoint is called the input, which is the exchange symmetry of the ansatz made manifest. Second, $P$ is doubly stochastic---its columns sum to one as well as its rows, by the dual orthogonality relation---so the uniform distribution is stationary. Third, the first reflection identity in Eq.~\eqref{eq:reflections} gives the particle-hole symmetry $P(\tilde 1_b\mid n,\tilde 1_a)=P(n-\tilde 1_b\mid n,n-\tilde 1_a)$ directly, confirming that the forced weighting realizes the relabeling invariance (C) in full.

Equation~\eqref{P in terms of Krawtchouk} is a combinatorial analog of the Born rule: the normalized amplitude
\begin{align}
\phi_{\tilde{1}_a}(\tilde{1}_b)\equiv \frac{K_{\tilde{1}_a}(\tilde{1}_b;n)}{\sqrt{\binom{n}{\tilde{1}_a}2^n}}\sqrt{\binom{n}{\tilde{1}_b}}
\end{align}
satisfies $P=|\phi_{\tilde{1}_a}|^2$. The quadratic structure itself is already present in the bilinear ansatz of Eq.~\eqref{eq:probability} and is part of the modeling choice, not shown to be unique under postulates (I)--(III) alone. What the forced sign adds is non-trivial: the forward and backward ensembles $\Upsilon_{a}$ and $\Upsilon_b$, built from \emph{independent} conditioning sets in Eq.~\eqref{eq abs epsilon}, are not free but collapse, via duality, into rescaled copies of a single amplitude function---which is why the product of two different objects becomes the square of one. In the ground sector $\tilde 1_a=0$ the amplitude reduces to $\phi_0(\tilde 1_b)=\sqrt{\binom{n}{\tilde 1_b}/2^n}$ and $P$ is the binomial distribution---the discrete Gaussian ground state.

It is worth pausing on where, physically, the interference lives in this construction. The forced sign is a function of the overlap $t=f_{11}$, equivalently of the Hamming distance $d_{ab}=\tilde 1_a+\tilde 1_b-2t$: an integer-valued label that behaves as a quantum number of the \emph{transition}, not of either state. It is nonlocal in a precise and limited sense---its value requires joint knowledge of both the input and the output sequence, and cannot be reconstructed from either endpoint alone---and by postulate (II) it is hidden from the observer, who sees only the endpoint counts $(\tilde 1_a,\tilde 1_b)$. Every observable transition therefore proceeds through a coherent sum over the admissible values of this relational hidden variable, and Theorem~\ref{thm:rigidity} says that the sum has no choice but to alternate in sign. The nodes of the Krawtchouk (and, in the continuum limit, Hermite) distributions are exact cancellations between channels that differ \emph{only} in $d_{ab}$: the $n=4$ example of Appendix~\ref{app:n4} displays a node at $\tilde 1_b=2$ produced by two channels, $d_{ab}\in\{1,3\}$, of equal cardinality and opposite parity. This gives interference a transparent combinatorial reading: it is not an added wave-dynamical ingredient, but the unique consistent bookkeeping over an unobservable degree of freedom that belongs to the pair of states jointly. We stress that ``nonlocal'' here refers to this relational, two-endpoint character---the variable spans the initial and final states of a single transition---and carries no implication of superluminal signaling or Bell-type nonlocality between spacelike-separated systems~\cite{bell1964epr}; the statement concerns a single register.

$\Upsilon_a$ coincides with the $p$-number $p_{\tilde{1}_b}$ of the binary Hamming association scheme $H_2^n$~\cite{Delsarte1998AssociationSA}. Figure~\ref{fig:distributions} shows the resulting distributions for $n=200$ and the first four excitation levels, reproducing the nodal structure of the QHO energy eigenstates; the worked $n=4$ example in Appendix~\ref{app:n4} verifies every step by direct integer counting.

\begin{figure}[t]
    \centering
    \includegraphics[width=\columnwidth]{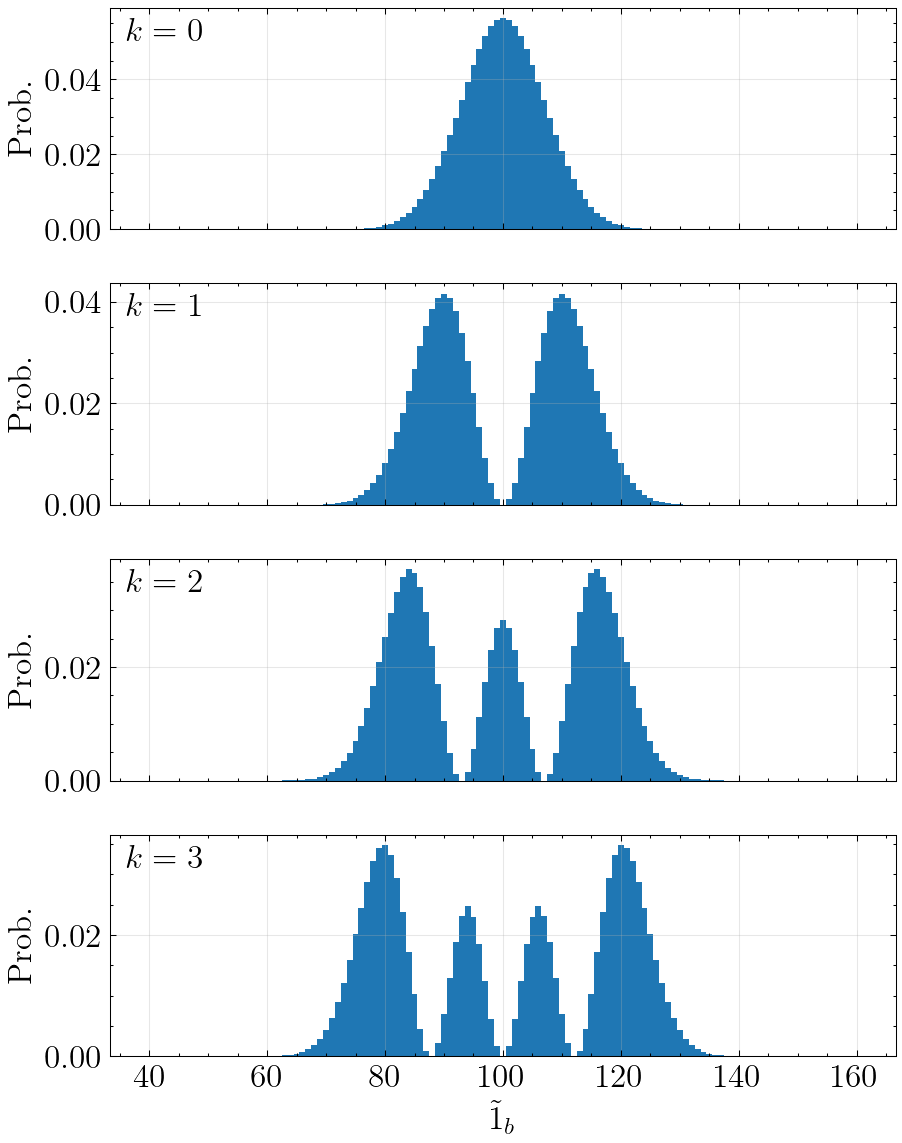}
    \caption{Probabilities $P$ [Eq.~\eqref{P in terms of Krawtchouk}] for $n=200$, for the first four excitation levels $\tilde{1}_a\in\{0,1,2,3\}$, reproducing the nodal structure of the QHO energy eigenstates. $\tilde{1}_a$ plays the role of the QHO energy level; the centered, rescaled output count $u=(n-2\tilde 1_b)/\sqrt{2n}$ plays the role of position.}
    \label{fig:distributions}
\end{figure}

\subsubsection{Discrete eigenvalue problem}

Writing $\psi_k(x)\equiv K_k(x;n)$, the difference equation~\eqref{difference equation} is a discrete eigenvalue problem,
\begin{align}
    \label{eigval problem}
    (n-x)\psi_k(x+1) + x\psi_k(x-1) = (n-2k)\psi_k(x),
\end{align}
i.e. $\hat{H}_{\rm disc}\psi_k = E_k\psi_k$ with
\begin{align}
    \label{hamiltonian and eigenvalue}
    \hat{H}_{\rm disc} = (n-x)T^+ + xT^-, \qquad E_k = n-2k,
\end{align}
where $T^\pm f(x) = f(x\pm1)$ are unit-shift operators. The spectrum $E_k=n,n-2,n-4,\ldots,-n$ is exactly equally spaced---the hallmark of a harmonic oscillator---with $k\equiv\tilde{1}_a$ playing the role of the principal quantum number. Two conventions deserve a remark. First, $E_k$ \emph{decreases} with $k$; the operator with the standard orientation and ground-state energy is obtained by the trivial affine shift $\hat H_{\rm osc}\equiv\tfrac12\big(n-\hat H_{\rm disc}\big)$, with spectrum $\hat H_{\rm osc}\psi_k = k\,\psi_k$, matching the QHO ladder $k=0,1,2,\ldots$ up to the zero-point constant. Second, $\hat H_{\rm disc}$ is not symmetric on plain $\ell^2(\{0,\ldots,n\})$; it is self-adjoint with respect to the binomial weight $\binom{n}{x}2^{-n}$, equivalently it becomes a real symmetric (Jacobi) matrix after the similarity transform $\psi\to\sqrt{\binom{n}{x}}\,\psi$ that produces the normalized functions used below. Its eigenvalues are therefore real and its eigenfunctions orthogonal in exactly the inner product with respect to which Eq.~\eqref{orthogonality} is an orthogonality statement.

\subsubsection{Continuum limit}

Introduce the centered, rescaled variable
\begin{align}
\label{eq:xi def}
\xi\equiv\frac{n-2x}{\sqrt{n}},
\end{align}
so that a unit shift $x\to x\pm1$ corresponds to $\xi\to\xi\mp\delta\xi$ with the full step $\delta\xi=2/\sqrt{n}\to0$ as $n\to\infty$. Taylor-expanding through second order, $\psi(x\pm1)=\psi\mp\psi'\,\delta\xi+\tfrac12\psi''\,\delta\xi^2\pm\mathcal{O}(\delta\xi^3)$, the two sides of Eq.~\eqref{eigval problem} combine as
\begin{align}
\label{TAYLOR EQ}
n\psi - (n-2x)\,\delta\xi\,\psi' + \tfrac{n}{2}\delta\xi^{2}\psi''
 + \mathcal{O}(n^{-1/2}) = (n{-}2k)\psi .
\end{align}
Using $n-2x=\xi\sqrt n$ and $\delta\xi=2/\sqrt n$, the surviving terms at $\mathcal{O}(1)$ give
\begin{align}
    \label{diff eq}
    \psi''(\xi) - \xi\psi'(\xi) + k\psi(\xi) = 0,
\end{align}
with the neglected terms of relative order $\mathcal{O}(1/n)$ [the cubic term carries $(n-2x)\,\delta\xi^3\sim\xi/n$ and the quartic term $n\,\delta\xi^4\sim1/n$], producing the $\mathcal{O}(1/n)$ corrections to the probability distributions that we check numerically in Sec.~\ref{sec:numerics}. Under the further rescaling $\xi=\sqrt{2}\,u$, i.e.
\begin{align}
\label{eq:u def}
u = \frac{n-2x}{\sqrt{2n}},
\end{align}
Eq.~\eqref{diff eq} is the Hermite differential equation~\cite{arfken2013mathematical}
\begin{align}
    \label{hermite eq}
    \psi''(u) - 2u\psi'(u) + 2k\psi(u) = 0,
\end{align}
with polynomial solutions $H_k(u)$. Writing $\psi(u)=\Phi(u)e^{+u^2/2}$,
\begin{align}
    \label{oscillator}
    -\frac{d^2\Phi}{du^2} + u^2\Phi = (2k+1)\Phi,
\end{align}
the time-independent Schr\"{o}dinger equation for the QHO in natural units, with $\Phi_k(u)\propto H_k(u)e^{-u^2/2}$ the familiar eigenfunctions and $u$ the centered, rescaled continuum limit of the output count, Eq.~\eqref{eq:u def} with $x=\tilde 1_b$. We use the single variable $u$ of Eq.~\eqref{eq:u def} in everything that follows, so that the difference-equation route and the recurrence route below land on the same Hermite functions with the same orientation.

Why does the parity sign, forced at finite $n$ to keep the effective cardinality pinned to $2^n$, survive in the limit $n\to\infty$ where no finite bound applies? The two regimes enforce the same mechanism in different guises: at finite $n$ the sign alternation cancels indistinguishable contributions so the observable cardinality never exceeds $2^n$; in the continuum limit the absolute bound is rescaled away, but the alternation, which steps once per unit change in $f_{11}$, smooths into the oscillatory nodal structure of $\Phi_k(u)$. Interference is not separately imposed in the limit; it is the smooth large-$n$ realization of the same constraint that was mandatory at every finite $n$.

The three-term recurrence Eq.~\eqref{3 term recurrence} gives a complementary view. Introducing the normalized Krawtchouk function $\phi_k\equiv K_k/\sqrt{\binom{n}{k}2^n}$, the exact recurrence reads
\begin{align}
\label{eq:exact-recurrence}
\sqrt{(k+1)(n-k)}\,\phi_{k+1} ={}& (n-2x)\,\phi_k \nonumber\\
&- \sqrt{k(n-k+1)}\,\phi_{k-1},
\end{align}
and dividing by $\sqrt{n}$ and using Eq.~\eqref{eq:u def}, it converges for $k\ll n$ to
\begin{align}
    \label{asymtpotic 3 term recurrence}
    \sqrt{k+1}\,\phi_{k+1}(u) \approx \sqrt{2}\,u\,\phi_k(u) - \sqrt{k}\,\phi_{k-1}(u),
\end{align}
which is precisely the recurrence of the normalized Hermite functions, equivalent to $H_{k+1}=2uH_k-2kH_{k-1}$~\cite{arfken2013mathematical}. The raising and lowering operators of the QHO thus emerge from the shift structure of the Krawtchouk recurrence in the large-$n$ limit. In summary,
\begin{align}
    \frac{\Upsilon_a\Upsilon_b}{2^n} = \frac{K_{\tilde{1}_b}(\tilde{1}_a;n)\,K_{\tilde{1}_a}(\tilde{1}_b;n)}{2^n} \xrightarrow{\;n\to\infty\;} |\Phi_{\tilde{1}_a}(u)|^2\,\Delta u,
\end{align}
with $\Delta u = 2/\sqrt{2n}$ the lattice spacing in $u$. This limit holds at fixed excitation $k$ as $n\to\infty$ (so $k/n\to0$); it is \emph{not} a statement that holds merely for large $n$ at fixed ratio $k/n$, and the two regimes behave differently, as Fig.~\ref{fig:scaling} shows directly.

\section{Numerical validation}
\label{sec:numerics}

Figure~\ref{fig:distributions} shows the probability distributions for $n=200$, reproducing the $k$-node pattern of the corresponding QHO eigenstates, and the worked $n=4$ example (Appendix~\ref{app:n4}) verifies every number by direct integer counting. The resulting distribution is exactly particle-hole symmetric, $P(\tilde{1}_b\mid n,\tilde{1}_a) = P(n-\tilde{1}_b\mid n,n-\tilde{1}_a)$, and this feature is not new: a bounded, mirror-symmetric spectrum is the defining property of the Krawtchouk finite oscillator, known since Ref.~\cite{Atakishiyev2005}. In the present construction it is one of the two explicit hypotheses that force the sign in the first place [condition (C) of Theorem~\ref{thm:rigidity}], alongside normalization. What the rigidity theorem adds is the demonstration that these two requirements together leave no freedom beyond the forced weighting.

\begin{figure}[t]
    \centering
    \includegraphics[width=0.95\columnwidth]{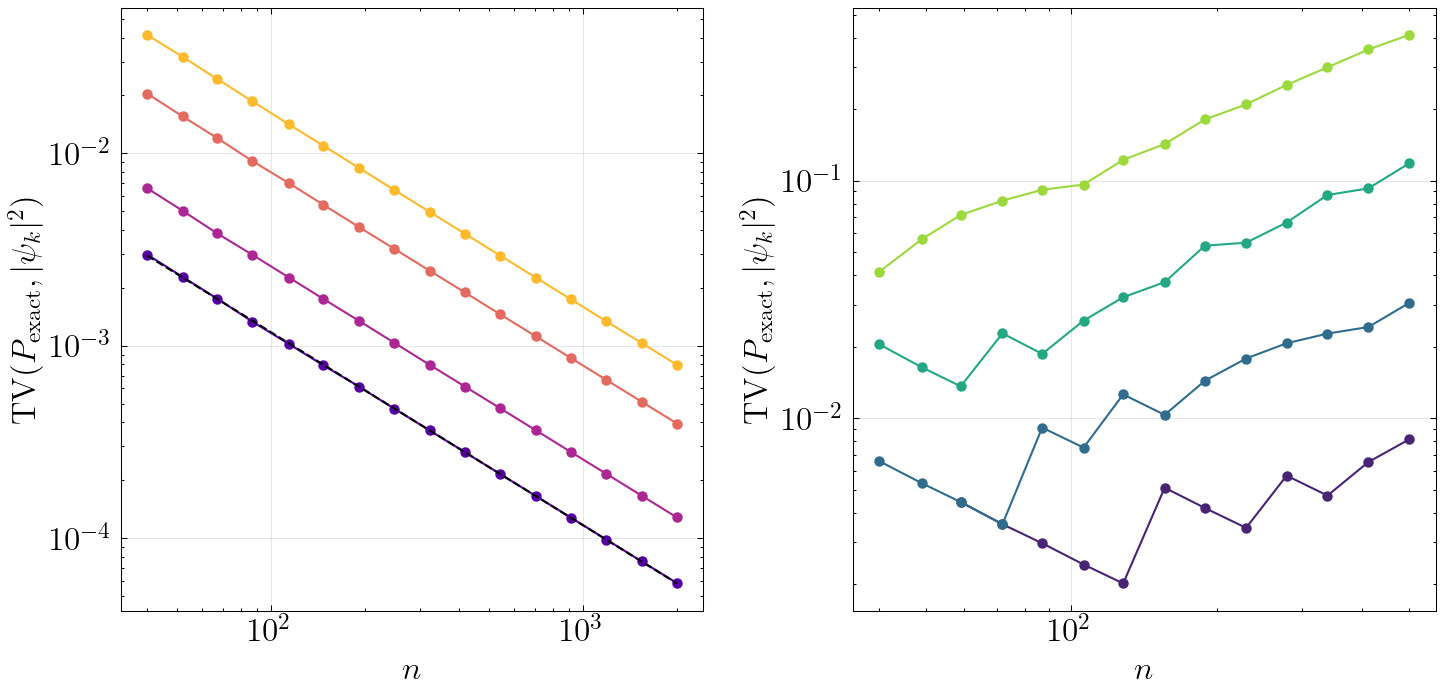}
    \caption{(a) At fixed excitation level $k$ and increasing $n$, the total-variation distance between the exact distribution [Eq.~\eqref{P in terms of Krawtchouk}] and the continuum Hermite-function prediction of Sec.~\ref{sec:krawtchouk} decays as a clean power law; a log-log fit gives exponent $\approx-1.00$, consistent with the $\mathcal{O}(1/n)$ error terms of Eq.~\eqref{TAYLOR EQ}. (b) At fixed ratio $k/n$ and increasing $n$, the same distance instead grows: the continuum limit requires $k$ fixed and $n\to\infty$ ($k/n\to0$), not merely $n$ large at fixed ratio.}
    \label{fig:scaling}
\end{figure}

The $n\to\infty$ contraction of the finite oscillator to the ordinary QHO, and its convergence rate, are also already established: the contraction is worked out directly in Ref.~\cite{Atakishiyev2003contraction} and sharpened with an explicit second-order rate in Ref.~\cite{chauleur2024kravchuk}. Figure~\ref{fig:scaling} is a numerical check that our derivation reproduces this, not a new result about the finite oscillator itself. Panel (a) fixes $k\in\{0,1,2,3\}$ and sweeps $n$ from $40$ to $2000$: the total-variation distance to the continuum Hermite-function prediction decays as a clean power law with fitted exponent $\approx -1.00$, matching the $\mathcal{O}(1/n)$ error estimate below Eq.~\eqref{diff eq}. Panel (b) fixes the ratio $k/n\in\{0.01,0.02,0.04,0.08\}$ instead and sweeps $n$: the deviation grows with $n$, because the excitation level $k=({\rm ratio})\times n$ grows without bound and the $k/n\ll1$ hypothesis behind the Taylor expansion of Eq.~\eqref{TAYLOR EQ} is violated by construction. The continuum limit is a two-parameter statement---an expansion in $k/n$ at fixed $k$, not simply an $n\to\infty$ limit---consistent with the standard asymptotic regimes of classical discrete orthogonal polynomials~\cite{Nikiforov_1991:inb}.

\begin{figure}[t]
    \centering
    \includegraphics[width=0.75\columnwidth]{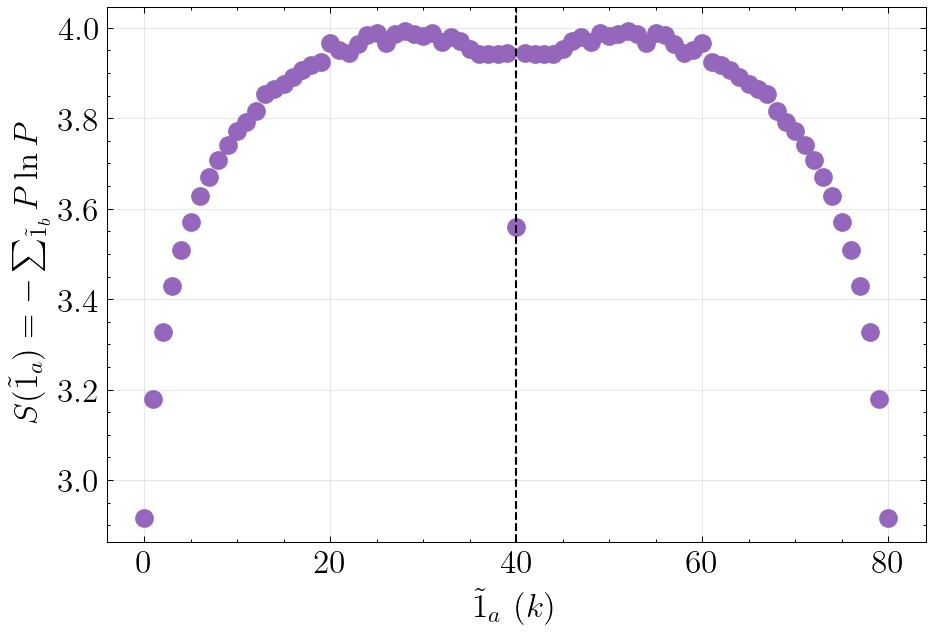}
    \caption{Shannon entropy of the transition measure $P(\tilde{1}_b\mid n,\tilde{1}_a)$ as a function of excitation level, for $n=80$. The sharp entropy dip at the particle-hole self-dual point $\tilde{1}_a=n/2$ has an exact explanation: the reflection identity $K_{n-k}(x;n)=(-1)^xK_k(x;n)$ forces $K_{n/2}(x;n)=0$ at every odd $x$, halving the support of $P$ at the self-dual level.}
    \label{fig:entropy}
\end{figure}

Figure~\ref{fig:entropy} shows the Shannon entropy $S(\tilde{1}_a) = -\sum_{\tilde{1}_b}P\ln P$ of the transition measure as a function of excitation level for $n=80$. The entropy rises from the ground state toward the bulk of the spectrum, as expected, but dips sharply at the particle-hole self-dual point $\tilde{1}_a=n/2$. This feature has an exact and elementary explanation, which we record because it appears not to have been stated in this form: the second reflection identity in Eq.~\eqref{eq:reflections}, $K_{n-k}(x;n)=(-1)^x K_k(x;n)$, evaluated at the self-dual level $k=n/2$ (for even $n$), gives $K_{n/2}(x;n) = (-1)^x K_{n/2}(x;n)$ and hence $K_{n/2}(x;n)=0$ for every odd $x$. The self-dual transition measure is therefore supported only on even output counts: exactly half of the outcomes are extinguished by a parity selection rule, and the entropy correspondingly drops by an amount of order $\ln 2$ relative to a smooth interpolation of the neighboring levels. (The $n=4$ table in Appendix~\ref{app:n4} displays this in miniature: the row $\tilde 1_a=2$ reads $\{3/8,\,0,\,1/4,\,0,\,3/8\}$.) The same dip independently corroborates a structurally analogous result in the companion exterior-algebra treatment~\cite{skrgic2026emergent}, which derives an entropy minimum at the Hodge self-dual grade of $\Lambda^\ast(\mathbb{R}^n)$ from the action of the Hodge star---a construction with no combinatorial reference to Hamming distances at all. Two independent constructions landing on an entropy minimum at the same self-dual point, both traceable to a self-duality selection rule, is evidence that the effect is a robust consequence of the underlying particle-hole-symmetric structure rather than an artifact of either framing.

The $\mathcal{O}(1/n)$ correction and particle-hole symmetry derived here are properties of the Krawtchouk finite oscillator, and systems with a bounded, symmetric-subspace structure---an $n$-qubit Dicke manifold~\cite{Dicke1954CoherenceIS}, or a bosonic mode truncated at $n$ excitations---are exactly where this model already lives: the $\mathfrak{su}(2)$ finite-oscillator literature discussed above is built on precisely this identification. What our derivation adds is a reason such a system should show Krawtchouk statistics \emph{specifically}, as a matter of information-theoretic necessity given a symbol-count restriction, a capacity bound, and the bilinear ansatz, rather than as one convenient finite-dimensional model among others.

This suggests a concrete experimental outlook. The QHO is realized to high precision in many controlled laboratory systems---the motional modes of trapped ions~\cite{RevModPhys.75.281}, microwave resonators in circuit QED~\cite{RevModPhys.93.025005}, phonon-number-resolving mechanical modes~\cite{Ding:2015xyg}, and collective-spin (Dicke) manifolds of $n$ qubits---and every such system has finite information capacity. The finite-$n$ deviations derived here are therefore, in principle, testable, and the predictions are parameter-free once $n$ is identified: the eigenstate quadrature statistics should follow the Krawtchouk form of Eq.~\eqref{P in terms of Krawtchouk} rather than the Hermite form, with total-variation deviations decaying as $c_k/n$ at fixed excitation $k$ [Fig.~\ref{fig:scaling}(a)] and coefficients $c_k$ that grow with $k$, so that moderately excited states at moderate $n$ maximize the signal; the spectrum should terminate after exactly $n+1$ equally spaced levels with exact particle-hole mirror symmetry; and at the self-dual excitation $k=n/2$ the outcome distribution should exhibit the parity selection rule of Fig.~\ref{fig:entropy}, with every odd output extinguished. In collective-spin realizations $n$ is directly the number of constituents and is known without fitting, which makes the $1/n$ scaling a sharp, falsifiable target. Three of these signatures admit exact closed forms at every $n$ and $k$, displayed in Fig.~\ref{fig:signatures}. The ground state is the binomial distribution, so its excess kurtosis is $\kappa_0=-2/n$ exactly. For the $k$-th level, squaring the exact recurrence Eq.~\eqref{eq:exact-recurrence} and using orthonormality gives $\langle(n-2x)^2\rangle_k=(k+1)(n-k)+k(n-k+1)=n(2k+1)-2k^2$, hence
\begin{align}
\label{eq:signatures}
\kappa_0=-\frac{2}{n},
\qquad
1-\frac{\langle u^2\rangle_k}{k+\tfrac12}=\frac{2k^2}{(2k+1)\,n},
\end{align}
with $\langle u^2\rangle_k = k+\tfrac12-k^2/n$ the exact position variance of level $k$. Finally, the support is hard-bounded: $P$ vanishes identically for $|u|>\sqrt{n/2}$, a cutoff no continuum eigenstate possesses. Each is a concrete target---a few-percent measurement of the ground-state kurtosis, or of the level-resolved variance deficit, determines the operative $n$ directly, giving quantitative content to the capacity measurement proposed below.

\begin{figure*}[t]
    \centering
    \includegraphics[width=0.98\textwidth]{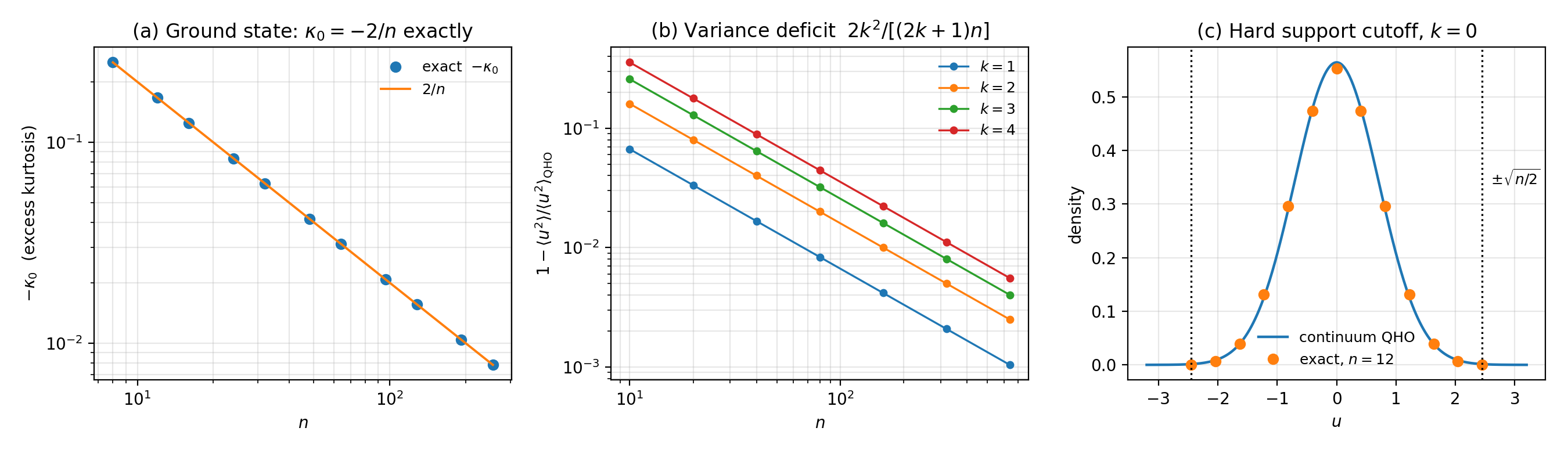}
    \caption{Exact closed-form finite-$n$ signatures [Eq.~\eqref{eq:signatures}] available to experiment. (a) Excess kurtosis of the ground-state position distribution: the exact value (dots) lies on $\kappa_0=-2/n$ (line) at every $n$. (b) Relative variance deficit of level $k$, $1-\langle u^2\rangle_k/(k+\tfrac12)=2k^2/[(2k+1)n]$, exact at every point shown. (c) The hard support cutoff: the exact $n=12$ ground state (dots) against the continuum QHO density (line); the distribution vanishes identically beyond $|u|=\sqrt{n/2}$ (dotted).}
    \label{fig:signatures}
\end{figure*} Two caveats delimit the claim. First, in the $\mathfrak{su}(2)$ realizations these deviations coincide with those of the standard finite-oscillator model~\cite{Atakishiyev2005,Atakishiyev2003contraction}, so such experiments test the Krawtchouk form itself rather than uniquely our derivation of it; the distinctive content of the present framework is the universality claim---that \emph{any} system whose information capacity is genuinely bounded by $2^n$ should exhibit Krawtchouk, rather than some other, finite-size deviations. Second, for systems whose capacity is only effectively truncated (a weakly anharmonic mode, a finite trap depth), identifying the operative $n$ is itself part of the experimental question; an observed family of deviations of the Krawtchouk form, scaling as $1/n$, would then function as a \emph{measurement} of that capacity.

\section{Discussion}
\label{sec:discussion}

We have shown that the Hamming/Delsarte association scheme, forced on the counting problem of binary sequences by an observability restriction to symbol counts, admits---once an exchange-symmetric bilinear transition measure is adopted---exactly one weighting of the hidden overlap consistent with symbol-relabeling symmetry and the $2^n$ capacity bound, across every transition sector at once. That weighting is the alternating sign, every alternative strictly overshoots the capacity bound by the computable amount of Eq.~\eqref{eq:overshoot-formula}, and the resulting measure is exactly the Krawtchouk finite oscillator, converging to the quantum harmonic oscillator in the large-$n$, fixed-excitation continuum limit. Interference, quadratic probabilities, and a discrete equally spaced spectrum are, in this precise and delimited sense, consequences of an information-capacity constraint on a combinatorial scheme together with an exchange-symmetric measure, rather than independent postulates layered on top of a Hilbert-space assumption. Because the derivation is exact at every finite $n$, it comes with finite-size predictions (Sec.~\ref{sec:numerics}) that are in principle testable wherever the oscillator is realized in a system of bounded information capacity.

The derivation also localizes the \emph{origin} of interference in a way the standard formalism does not. In ordinary quantum mechanics, the oscillating signs of the position-space eigenfunctions $\langle x|k\rangle$ are outputs of the formalism with no further explanation inside it. Here the same oscillation is traced to a single cause: the sign is carried by $d_{ab}$ (equivalently $t$), a hidden quantum number of the transition as a whole---defined jointly by the initial and the final state, invisible from either alone---and the alternating weighting over it is the unique bookkeeping consistent with symbol-relabeling symmetry and the capacity bound. Interference, in this model, is what counting looks like when an observable process must be marginalized over a relational variable that no measurement of a single endpoint can access; the forward and backward ensembles of Eq.~\eqref{eq:bilinear-ansatz}, which use both endpoints symmetrically, are the structural expression of this two-endpoint character---and, as argued in Sec.~\ref{sec:rigidity}, of the requirement that neither the preparation nor the outcome be specified beyond its observable counts. To the extent that the construction reproduces the QHO exactly, it suggests a reading of quantum interference more generally as constrained bookkeeping over hidden variables attached to \emph{pairs} (preparation, outcome) rather than to states---while remaining, as emphasized above, entirely distinct from Bell-type nonlocality between separated systems~\cite{bell1964epr}.

In every construction of the finite oscillator we are aware of, the logical order is to postulate a kinematical Lie algebra---$\mathfrak{su}(2)$~\cite{Atakishiyev2005,Atakishiyev_2001:JPA}, a $\mathcal{CP}$-deformed variant of $\mathfrak{su}(2)$~\cite{jafarov2015cpkrawtchouk}, or an $\mathfrak{so}(d+1)$ representation acting on multivariate Krawtchouk polynomials~\cite{genest2013multivariate,genest2015charlier}---and then identify the Krawtchouk polynomials as eigenfunctions of a resulting finite-difference Hamiltonian. Recent work continues this direction: Ref.~\cite{mayqin2024algebraic} builds an algebraic discrete QHO with an $\mathfrak{su}(2)$ kinematical algebra and dynamic resolution-scaling ladder operators. The present work inverts that logical order: no Lie algebra, representation, or Hilbert space is postulated anywhere in Sec.~\ref{sec:scheme}. The rigidity theorem shows instead that the $\mathfrak{su}(2)$-like ladder structure and its Krawtchouk eigenfunctions are the unique output of an information-capacity constraint on a combinatorial scheme, given the bilinear ansatz of Sec.~\ref{sec:rigidity}. What Sec.~\ref{sec:numerics} adds is not new facts about the finite oscillator, which is already well understood, but a check that the object our rigidity theorem forces us into is the same finite oscillator, not a different one that merely looks similar.

The sign structure derived here also has a geometric reading, developed in full in a companion manuscript~\cite{skrgic2026emergent}: promoting the bit structure to Grassmann generators $\theta_i$ with $\theta_i^2=0$ maps binary sequences to monomials in the exterior algebra $\Lambda^\ast(\mathbb{R}^n)$, whose dimension $2^n$ matches the sequence space exactly, and the forced weighting is precisely the Koszul sign~\cite{berger2017koszulsignmap} of the canonical inner product on $\Lambda^\ast(\mathbb{R}^n)$. That manuscript develops the connection independently, deriving an emergent $\mathfrak{su}(2)$ operator algebra, a Hodge-duality particle-hole symmetry, and a fermionic thermodynamic sector directly from the exterior-algebra structure and its own grade-observability restriction, with the QHO recovered as its own contraction limit. We keep the two manuscripts separate because they address different scopes---a transition-probability rigidity theorem here; an operator algebra and thermodynamics there---but the cross-check of Fig.~\ref{fig:entropy} shows the two constructions are independent enough to provide a genuine consistency check on each other, not merely two expositions of the same fact.

A natural question raised by the present work is the extension of the association-scheme argument beyond the binary alphabet, where preliminary indications~\cite{skrgic2026emergent} suggest a ternary generalization contracting to $\mathfrak{su}(3)$ with a $\mathbb{Z}_3$-graded structure. This ternary extension is an instance of a developing program in which finite alphabets and their relabeling groups generate the full Standard Model gauge group directly, with particles realized as words spelled from letters in these alphabets. The mechanism responsible for the sign rigidity derived here reappears as the $\mathbb{Z}_3$ center of the ternary register, which permits only triality-zero words and thereby enforces quark confinement as a selection rule rather than a postulate.

Finally, establishing that the bilinear ansatz itself is forced, rather than assumed, would require showing that every non-negative, normalized function built from the $|S^2|$ cardinalities factors through the Bose--Mesner algebra of the scheme. This is a natural but separate question, which we leave open.

\begin{acknowledgments}
D.S. and D.\v{S}. are partially supported by the US National Science Foundation, under Grant No. PHY-2310363. Claude and Copilot were used in creation of this manuscript, mainly to test various possible realizations of the framework, generate the concrete examples, numerical checks and improve the text.
\end{acknowledgments}

%% ------------------------------------------------------------------
\appendix

\section{Proof of the rigidity theorem}
\label{app:proof}

Throughout this appendix we abbreviate $a\equiv\tilde 1_a$ and $b\equiv\tilde 1_b$, use the convention $\binom{m}{j}=0$ outside $0\le j\le m$ so that all sums over $t$ run over $\mathbb{Z}$, and write
\begin{align}
\begin{split}
\Upsilon_a(a,b) &= \sum_{t} w(t)\binom{a}{t}\binom{n-a}{b-t},\\
\Upsilon_b(a,b) &= \sum_{t} w(t)\binom{b}{t}\binom{n-b}{a-t},
\end{split}
\end{align}
suppressing the fixed $n$. Note for later use that both sums involve $w(t)$ only for $0\le t\le\min(a,b)\le a$: the sector $a$ constrains only the values $w(0),\ldots,w(a)$.

We will use two elementary lemmas.

\begin{lemma}[Binomial transform of a Krawtchouk polynomial]
\label{lem:gf}
For any $0\le m\le n$ and $0\le k\le n$,
\begin{align}
\sum_{x=0}^{m}\binom{m}{x}K_k(x;n) = 2^m\binom{n-m}{k}.
\end{align}
In particular, for $m=n-a$ and $k=a$, $\sum_{x}\binom{n-a}{x}K_a(x;n) = 2^{n-a}$.
\end{lemma}

\begin{proof}
From the generating function $\sum_k K_k(x;n)z^k = (1-z)^x(1+z)^{n-x}$,
\begin{align}
\sum_{x}\binom{m}{x}K_k(x;n)
&= [z^k]\,(1+z)^n\sum_x\binom{m}{x}\Big(\frac{1-z}{1+z}\Big)^{x} \nonumber\\
&= [z^k]\,(1+z)^n\Big(1+\frac{1-z}{1+z}\Big)^{m} \nonumber\\
&= 2^m\,[z^k]\,(1+z)^{n-m},
\end{align}
which is $2^m\binom{n-m}{k}$. For $m=n-a$, $k=a$: $2^{n-a}\binom{a}{a}=2^{n-a}$.
\end{proof}

\begin{lemma}[Three sector sums]
\label{lem:sums}
For $1\le a\le n-1$, define
\begin{align}
\begin{split}
S_1 &\equiv \sum_{b} K_b(a;n)\binom{b}{a},\\
S_2 &\equiv \sum_{b} K_a(b;n)\binom{n-a}{b-a},\\
S_3 &\equiv \sum_{b} \binom{n-a}{b-a}\binom{b}{a}.
\end{split}
\end{align}
Then
\begin{align}
\begin{split}
S_1 &= S_2 = (-1)^a\,2^{\,n-a},\\
S_3 &= \sum_{j\ge0}\binom{a}{j}\binom{n-a}{j}\,2^{\,n-a-j} \;>\; 2^{\,n-a}.
\end{split}
\end{align}
\end{lemma}

\begin{proof}
\emph{$S_1=S_2$:} by the duality relation Eq.~\eqref{eq:duality}, $K_b(a;n)=K_a(b;n)\binom{n}{b}/\binom{n}{a}$, and by the trinomial revision $\binom{n}{b}\binom{b}{a}=\binom{n}{a}\binom{n-a}{b-a}$; substituting both into $S_1$ gives $S_2$.

\emph{$S_2=(-1)^a2^{n-a}$:} the first reflection identity in Eq.~\eqref{eq:reflections} gives $K_a(b;n)=(-1)^aK_a(n-b;n)$; substituting $x=n-b$ and using $\binom{n-a}{b-a}=\binom{n-a}{(n-a)-x}=\binom{n-a}{x}$,
\begin{align}
S_2 = (-1)^a\sum_{x}\binom{n-a}{x}K_a(x;n) = (-1)^a\,2^{\,n-a}
\end{align}
by Lemma~\ref{lem:gf}.

\emph{$S_3$:} write $b=a+c$ and use the Vandermonde form $\binom{a+c}{a}=\sum_j\binom{a}{j}\binom{c}{c-j}=\sum_j\binom{a}{j}\binom{c}{j}$, then $\sum_c\binom{n-a}{c}\binom{c}{j}=\binom{n-a}{j}2^{\,n-a-j}$ (choose the $j$ marked elements first, the rest freely). The $j=0$ term alone is $2^{n-a}$, and the $j=1$ term, $a(n-a)2^{\,n-a-1}$, is strictly positive for $1\le a\le n-1$; hence $S_3>2^{n-a}$.
\end{proof}

\emph{Proof of Theorem~\ref{thm:rigidity}.} The overall sign of $w$ flips $\Upsilon_a$ and $\Upsilon_b$ simultaneously and cancels in $P$; we may therefore fix $w(0)>0$ once Step~1 shows $w(0)\neq0$, and reinstate the overall sign $\sigma$ at the end.

\emph{Step 1: $w(0)^2=1$.} In the sector $a=0$, only the channel $t=0$ is admissible, so $\Upsilon_a(0,b)=w(0)\binom{n}{b}$ and $\Upsilon_b(0,b)=w(0)$. Hypothesis (N) gives
\begin{align}
\sum_b P(b\mid n,0) = \frac{w(0)^2}{2^n}\sum_b\binom{n}{b} = w(0)^2 = 1 .
\end{align}
Set $w(0)=+1$.

\emph{Step 2: $w(t)^2=1$ for every $t$.} In the sector $a=n$, only the channel $t=b$ is admissible [$\binom{n}{t}\binom{0}{b-t}$ forces $t=b$, and $\binom{b}{t}\binom{n-b}{n-t}$ forces the same], so
\begin{align}
\Upsilon_a(n,b)=w(b)\binom{n}{b},\qquad \Upsilon_b(n,b)=w(b),
\end{align}
and hence $P(b\mid n,n)=w(b)^2\binom{n}{b}/2^n$. Hypothesis (C) applied to the sector pair $(0,b)\leftrightarrow(n,n-b)$ requires
\begin{align}
P(n-b\mid n,n) = \frac{w(n-b)^2\binom{n}{n-b}}{2^n} \overset{!}{=} P(b\mid n,0)=\frac{\binom{n}{b}}{2^n},
\end{align}
i.e. $w(n-b)^2=1$ for every $b$, hence $w(t)\in\{\pm1\}$ for every $t\in\{0,\ldots,n\}$.

\emph{Step 3: induction on the sector.} Suppose $w(t)=(-1)^t$ for all $t<a$, for some $1\le a\le n-1$; Step 1 provides the base case. Since sector $a$ involves only $w(0),\ldots,w(a)$, and $w(a)=\pm(-1)^a$ by Step 2, write $w(a)=(-1)^a+\delta$ with $\delta\in\{0,\;-2(-1)^a\}$. Then, isolating the $t=a$ terms,
\begin{align}
\Upsilon_a(a,b) &= K_b(a;n) + \delta\binom{n-a}{b-a}, \nonumber\\
\Upsilon_b(a,b) &= K_a(b;n) + \delta\binom{b}{a},
\end{align}
using $\binom{a}{a}=\binom{n-b}{a-a}=1$ [wherever these deviate from zero the constraint $b\ge a$ is automatic through the binomials]. Hypothesis (N) in sector $a$, together with $\sum_bK_b(a;n)K_a(b;n)=2^n$ [orthogonality, Eq.~\eqref{orthogonality}, after duality], gives
\begin{align}
\sum_b\Upsilon_a\Upsilon_b
&= 2^n + \delta\,(S_1+S_2) + \delta^2 S_3 \nonumber\\
&= 2^n + 2\delta(-1)^a2^{\,n-a} + \delta^2 S_3
\end{align}
by Lemma~\ref{lem:sums}. For the correct sign, $\delta=0$, the sum is $2^n$ and (N) holds. For the wrong sign, $\delta=-2(-1)^a$, the sum is
\begin{align}
2^n - 4\cdot2^{\,n-a} + 4S_3 = 2^n + 4\big(S_3-2^{\,n-a}\big) > 2^n,
\end{align}
strictly, by Lemma~\ref{lem:sums}. The wrong sign therefore violates (N)---and does so by \emph{overshooting} the capacity bound, the same failure mode as the fully unsigned measure of Eq.~\eqref{eq:overshoot}. Hence $w(a)=(-1)^a$, completing the induction through $a=n-1$.

\emph{Step 4: the endpoint $t=n$.} The value $w(n)$ enters $\Upsilon_a$ or $\Upsilon_b$ only when $t=n$ is admissible, which requires $a=b=n$; there $\Upsilon_a(n,n)=\Upsilon_b(n,n)=w(n)$ and $P(n\mid n,n)=w(n)^2/2^n$. Thus $w(n)$ enters $P$ only quadratically, its sign is unobservable, and Step 2 fixes $w(n)^2=1$; both hypotheses are satisfied by either sign, consistent with the statement $w(n)=\pm(-1)^n$. [One checks that (N) holds in sector $n$ automatically: $\sum_bP(b\mid n,n)=\sum_bw(b)^2\binom{n}{b}/2^n=1$ by Step 2.]

\emph{Existence.} It remains to verify that the forced weighting satisfies (C) in full, not only on the slice used in Step 2. With $w(t)=(-1)^t$, Sec.~\ref{sec:krawtchouk} gives $P=\binom{n}{b}\binom{n}{a}^{-1}K_a(b;n)^2/2^n$; the reflection identities Eq.~\eqref{eq:reflections} give
\begin{align}
K_{n-a}(n-b;n) &= (-1)^{n-a}\,K_{n-a}(b;n) \nonumber\\
&= (-1)^{n-a}(-1)^{b}\,K_a(b;n),
\end{align}
which squares to $K_a(b;n)^2$, and $\binom{n}{n-b}\big/\binom{n}{n-a}=\binom{n}{b}\big/\binom{n}{a}$; hence $P(n-b\mid n,n-a)=P(b\mid n,a)$, i.e. Eq.~\eqref{eq:relabeling} holds identically. Non-negativity is manifest from the squared form. \hfill$\blacksquare$

\begin{remark}[Normalization alone is not rigid]
\label{rem:spurious}
Hypothesis (C) cannot be dropped. Impose (N) alone in the sectors $a=0$ and $a=1$, with $w(0)=1$ and $w(1)=w_1$ free. Sector $a=1$ has $\Upsilon_a(1,b)=\binom{n-1}{b}+w_1\binom{n-1}{b-1}$ and $\Upsilon_b(1,b)=(n-b)+w_1 b$, and evaluating $\sum_b\Upsilon_a\Upsilon_b=2^n$ term by term yields the quadratic
\begin{align}
(n+1)\,w_1^2 + 2(n-1)\,w_1 + (n-3) = 0,
\end{align}
whose roots are $w_1=-1$ and $w_1=-(n-3)/(n+1)$. Normalization at the first excited sector therefore admits a spurious branch with $|w_1|\neq1$; it is the relabeling symmetry, through Step 2, that eliminates it and every branch like it.
\end{remark}

\section{Worked example: $n=4$}
\label{app:n4}

\begin{figure*}
\centering
\subfloat[]{\includegraphics[width=0.3\linewidth]{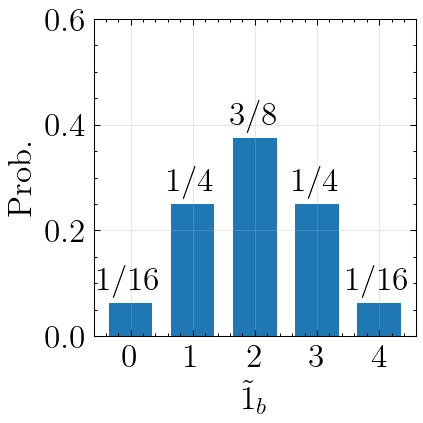}}
\subfloat[]{\includegraphics[width=0.3\linewidth]{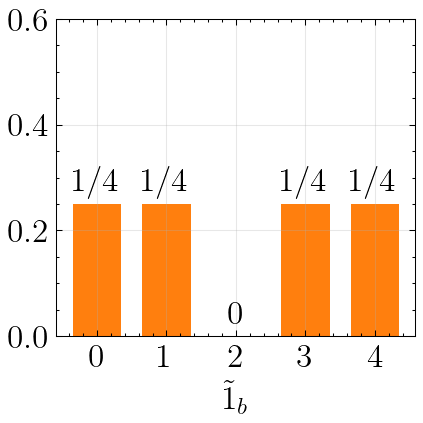}}
\subfloat[]{\includegraphics[width=0.3\linewidth]{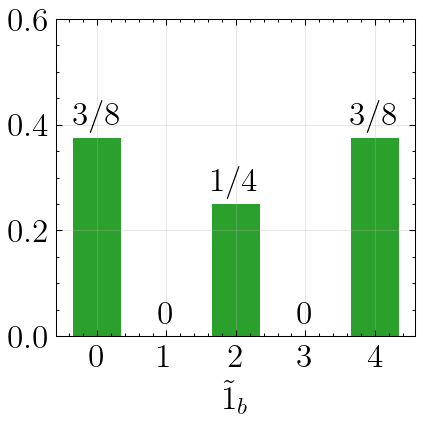}}

\vspace{0.3cm}

\subfloat[]{\includegraphics[width=0.3\linewidth]{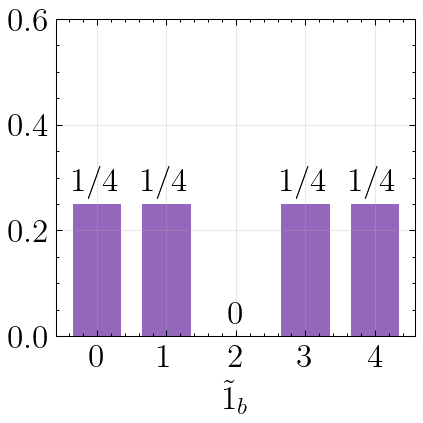}}
\subfloat[]{\includegraphics[width=0.3\linewidth]{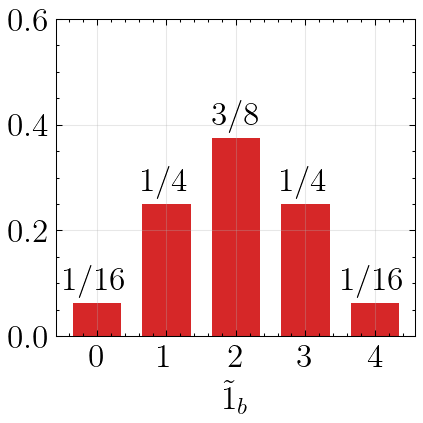}}

\caption{$P(\tilde{1}_b\mid n=4,\tilde{1}_a)$ for $\tilde{1}_a=0$ (a), $1$ (b), $2$ (c), $3$ (d) and $4$ (e). Exact rational values label each bar; zeros are annotated. The node count equals $\tilde{1}_a$ in every panel, and the particle-hole symmetry is visible by comparing $\tilde{1}_a=0\leftrightarrow4$ and $\tilde{1}_a=1\leftrightarrow3$.}
\label{fig:n4_bars}
\end{figure*}

We work through every computational step for $n=4$, $\tilde{1}_a=1$ (first excited level), then display the full probability table for all five levels $\tilde{1}_a=0,\ldots,4$. All numbers come from integer counting alone; the Krawtchouk formula~\eqref{P in terms of Krawtchouk} is checked against them at the end.

Fixing $n=4$, $\tilde{1}_a=1$: Alice's sequences are all length-4 strings with exactly one $1$. Equation~\eqref{s1} gives
\begin{align}
\begin{split}
  S^1(4,1) &= \{0001,0010,0100,1000\},\\
  |S^1(4,1)| &= \binom{4}{1}=4.
\end{split}
\end{align}

For each output count $\tilde{1}_b\in\{0,1,2,3,4\}$, the overlap $t=f_{11}$ runs over the admissible set $\Lambda$ of Eq.~\eqref{eq:Lambda}, obtained by requiring all four frequencies to be non-negative integers; equivalently, the Hamming distance $d_{ab}=\tilde 1_a+\tilde 1_b-2t$ runs over the mirror set. Both are listed here:
\begin{center}
\begin{tabular}{c|c|c}
$\tilde{1}_b$ & $\Lambda$ ($t$ values) & $d_{ab}$ values \\\hline
0 & $\{0\}$   & $\{1\}$ \\
1 & $\{0,1\}$ & $\{2,0\}$ \\
2 & $\{0,1\}$ & $\{3,1\}$ \\
3 & $\{0,1\}$ & $\{4,2\}$ \\
4 & $\{1\}$   & $\{3\}$ \\
\end{tabular}
\end{center}

Table~\ref{tab:app} records, for every $(\tilde{1}_b,d_{ab})$ pair, the base-4 frequencies, the measurement-set cardinality $|S^2|=4!/(f_{00}!f_{01}!f_{10}!f_{11}!)$, the parity sign $(-1)^{t}=(-1)^{f_{11}}$, and the two elementary cardinalities of Eq.~\eqref{eq:eps closed}.

\begin{table}
\centering
\renewcommand{\arraystretch}{1.25}
\setlength{\tabcolsep}{4pt}
\begin{tabular}{@{}cccccccccc@{}}
\hline\hline
$\tilde{1}_b$ & $d_{ab}$ & $f_{11}$ & $f_{10}$ & $f_{01}$ & $f_{00}$ & $|S^2|$ & $(-1)^{f_{11}}$ & $|\varepsilon_a|$ & $|\varepsilon_b|$\\
\hline
0 & 1 & 0 & 1 & 0 & 3 &  4 & $+1$ & 1 & 4 \\
\hline
1 & 0 & 1 & 0 & 0 & 3 &  4 & $-1$ & 1 & 1 \\
  & 2 & 0 & 1 & 1 & 2 & 12 & $+1$ & 3 & 3 \\
\hline
2 & 1 & 1 & 0 & 1 & 2 & 12 & $-1$ & 3 & 2 \\
  & 3 & 0 & 1 & 2 & 1 & 12 & $+1$ & 3 & 2 \\
\hline
3 & 2 & 1 & 0 & 2 & 1 & 12 & $-1$ & 3 & 3 \\
  & 4 & 0 & 1 & 3 & 0 &  4 & $+1$ & 1 & 1 \\
\hline
4 & 3 & 1 & 0 & 3 & 0 &  4 & $-1$ & 1 & 4 \\
\hline\hline
\end{tabular}
\caption{Base-4 frequencies, cardinalities, and parity signs for $n=4$, $\tilde{1}_a=1$.}
\label{tab:app}
\end{table}

Summing each column of Table~\ref{tab:app} with its parity weight gives the ensembles $\Upsilon_a$ and $\Upsilon_b$, and hence $P(\tilde{1}_b\mid4,1)=\Upsilon_a\Upsilon_b/2^4=\Upsilon_a\Upsilon_b/16$:
\begin{align}
\begin{array}{c|rrrrr}
\tilde{1}_b & 0 & 1 & 2 & 3 & 4\\\hline
\Upsilon_a  & 1 & 2 & 0 & -2 & -1\\
\Upsilon_b  & 4 & 2 & 0 & -2 & -4\\
\Upsilon_a\Upsilon_b & 4 & 4 & 0 & 4 & 4\\
P           & 1/4 & 1/4 & 0 & 1/4 & 1/4
\end{array}
\label{eq:app_prob_table}
\end{align}

Sum $=4\times\tfrac14=1$. The zero at $\tilde{1}_b=2$ is exactly the node from destructive cancellation: the two channels $t\in\{1,0\}$ (i.e. $d_{ab}\in\{1,3\}$) carry equal cardinalities ($|\varepsilon_a|=3$ each) but opposite signs. Note what distinguishes the two cancelling channels: nothing observable. They agree on every endpoint count and differ only in $d_{ab}$, the variable defined by the input and output jointly---the interference node is bookkeeping over precisely this nonlocal, hidden label. Without the parity sign, the unsigned bilinear measure of Eq.~\eqref{eq:naive} would give $P^{\rm cl}(2\mid4,1)=(3+3)(2+2)/16=24/16>1$, directly violating normalization; indeed the full unsigned row sums to $\binom{4}{1}=4$, the $n=4$ instance of the exact overshoot Eq.~\eqref{eq:overshoot} illustrated in Fig.~\ref{fig:capacity}(a).

Repeating for all $\tilde{1}_a$:
\begin{align}
\begin{array}{c|ccccc}
\tilde{1}_a\backslash\tilde{1}_b & 0 & 1 & 2 & 3 & 4\\\hline
0 & 1/16 & 1/4  & 3/8  & 1/4  & 1/16\\
1 & 1/4  & 1/4  & 0    & 1/4  & 1/4 \\
2 & 3/8  & 0    & 1/4  & 0    & 3/8 \\
3 & 1/4  & 1/4  & 0    & 1/4  & 1/4 \\
4 & 1/16 & 1/4  & 3/8  & 1/4  & 1/16
\end{array}
\label{eq:app_full_table}
\end{align}
Each row sums to $1$; by the symmetry of Eq.~\eqref{eq:P symmetric form} the table is symmetric and each column sums to $1$ as well (the transition matrix is doubly stochastic). The number of zeros in row $\tilde{1}_a$ equals $\tilde{1}_a$, matching the node count of the $k=\tilde{1}_a$ QHO eigenstate; particle-hole symmetry $P(\tilde{1}_b\mid n,\tilde{1}_a)=P(n-\tilde{1}_b\mid n,n-\tilde{1}_a)$ is evident (rows $0,4$ and $1,3$ are identical); and the parity selection rule at the self-dual level is visible in row $\tilde 1_a=2$, whose odd entries vanish identically, the $n=4$ instance of the entropy-dip mechanism of Sec.~\ref{sec:numerics}. Finally, marginalizing a \emph{single} unsigned ensemble over the hidden overlap, $\sum_t|\varepsilon_a|/2^n=\binom{4}{\tilde 1_b}/16=\{1,4,6,4,1\}/16$, reproduces the plain binomial distribution for every $\tilde{1}_a$: a valid probability in the abstract, but one that carries no memory of the input state and so cannot describe a physical transition.


\begin{thebibliography}{99}

\bibitem{sakurai2017modernqm}
J.~J. Sakurai and J.~Napolitano, \emph{Modern Quantum Mechanics}, 2nd~ed. (Cambridge University Press, Cambridge, 2017).

\bibitem{walls2008quantum}
D.~F. Walls and G.~J. Milburn, \emph{Quantum Optics} (Springer, Berlin, Heidelberg, 2008).

\bibitem{peskin1995introduction}
M.~E. Peskin and D.~V. Schroeder, \emph{An Introduction to Quantum Field Theory}, Frontiers in Physics (Avalon Publishing, 1995).

\bibitem{Ding:2015xyg}
S.~Ding, G.~Maslennikov, R.~Hablutzel, H.~Loh, and D.~Matsukevich, Cross-Kerr nonlinearity for phonon counting, Phys. Rev. Lett. \textbf{119}, 150404 (2017), arXiv:1512.01670.

\bibitem{RevModPhys.75.281}
D.~Leibfried, R.~Blatt, C.~Monroe, and D.~Wineland, Quantum dynamics of single trapped ions, Rev. Mod. Phys. \textbf{75}, 281 (2003).

\bibitem{RevModPhys.93.025005}
A.~Blais, A.~L. Grimsmo, S.~M. Girvin, and A.~Wallraff, Circuit quantum electrodynamics, Rev. Mod. Phys. \textbf{93}, 025005 (2021).

\bibitem{hardy2001quantumtheoryreasonableaxioms}
L.~Hardy, Quantum theory from five reasonable axioms, arXiv:quant-ph/0101012 (2001).

\bibitem{dariano2017}
G.~M. D'Ariano, Physics without physics, Int. J. Theor. Phys. \textbf{56}, 97 (2017).

\bibitem{chiribella2011informational}
G.~Chiribella, G.~M. D'Ariano, and P.~Perinotti, Informational derivation of quantum theory, Phys. Rev. A \textbf{84}, 012311 (2011), arXiv:1011.6451.

\bibitem{spekkens2007toy}
R.~W. Spekkens, Evidence for the epistemic view of quantum states: A toy theory, Phys. Rev. A \textbf{75}, 032110 (2007), arXiv:quant-ph/0401052.

\bibitem{catani2017spekkens}
L.~Catani and D.~E. Browne, Spekkens' toy model in all dimensions and its relationship with stabilizer quantum mechanics, New J. Phys. \textbf{19}, 073035 (2017), arXiv:1701.07801.

\bibitem{wetterich2010qmfromclassical}
C.~Wetterich, Quantum mechanics from classical statistics, Ann. Phys. (N.Y.) \textbf{325}, 852 (2010), arXiv:0906.4919.

\bibitem{thooft2016cellular}
G.~'t~Hooft, \emph{The Cellular Automaton Interpretation of Quantum Mechanics}, Fundamental Theories of Physics Vol.~185 (Springer International Publishing, 2016).

\bibitem{elze2014quantumness}
H.-T. Elze, Quantumness of discrete Hamiltonian cellular automata, EPJ Web Conf. \textbf{78}, 02005 (2014), arXiv:1407.2160.

\bibitem{Powers:2021rfg}
S.~Powers and D.~Stojkovic, An alternative formalism for modeling spin, Eur. Phys. J. C \textbf{82}, 690 (2022), arXiv:2110.13617.

\bibitem{Powers:2023ydg}
S.~Powers, G.~Xu, H.~Fotso, T.~Thomay, and D.~Stojkovic, Statistical model for quantum spin and photon number states, Phys. Rev. A \textbf{111}, 012217 (2025), arXiv:2304.13535.

\bibitem{Powers:2023sqf}
S.~Powers and D.~Stojkovic, An event centric approach to modeling quantum systems, arXiv:2306.14922 (2023).

\bibitem{coleman2011krawtchouk}
R.~Coleman, On Krawtchouk polynomials, arXiv:1101.1798 (2011).

\bibitem{Feinsilver:2007oyk}
P.~Feinsilver and J.~Kocik, Krawtchouk polynomials and Krawtchouk matrices, arXiv:quant-ph/0702073 (2007).

\bibitem{kravchuk1929generalisations}
M.~Kravchuk, Sur une g\'{e}n\'{e}ralisation des polyn\^{o}mes d'Hermite, C. R. Acad. Sci. Paris \textbf{189}, 620 (1929).

\bibitem{Atakishiyev2005}
N.~M. Atakishiyev, G.~S. Pogosyan, and K.~B. Wolf, Finite models of the oscillator, Phys. Part. Nucl. \textbf{36}, 247 (2005).

\bibitem{Atakishiyev_2001:JPA}
N.~M. Atakishiyev, G.~S. Pogosyan, L.~E. Vicent, and K.~B. Wolf, Finite two-dimensional oscillator. I. The Cartesian model, J. Phys. A \textbf{34}, 9381 (2001).

\bibitem{Nikiforov_1991:inb}
A.~F. Nikiforov, V.~B. Uvarov, and S.~K. Suslov, \emph{Classical Orthogonal Polynomials of a Discrete Variable} (Springer, Berlin, Heidelberg, 1991), pp.~18--54.

\bibitem{Delsarte1998AssociationSA}
P.~Delsarte and V.~I. Levenshtein, Association schemes and coding theory, IEEE Trans. Inf. Theory \textbf{44}, 2477 (1998).

\bibitem{bell1964epr}
J.~S. Bell, On the Einstein-Podolsky-Rosen paradox, Physics Physique Fizika \textbf{1}, 195 (1964).

\bibitem{arfken2013mathematical}
G.~B. Arfken, H.~J. Weber, and F.~E. Harris, \emph{Mathematical Methods for Physicists: A Comprehensive Guide} (Elsevier Science, 2013).

\bibitem{Atakishiyev2003contraction}
N.~M. Atakishiyev, G.~S. Pogosyan, and K.~B. Wolf, Contraction of the finite one-dimensional oscillator, Int. J. Mod. Phys. A \textbf{18}, 317 (2003).

\bibitem{chauleur2024kravchuk}
Q.~Chauleur and E.~Faou, Discrete quantum harmonic oscillator and Kravchuk transform, ESAIM: Math. Model. Numer. Anal. \textbf{58}, 2155 (2024), arXiv:2212.03164.

\bibitem{skrgic2026emergent}
D.~\v{S}krgi\'{c} and D.~Stojkovic, Observable Hilbert spaces from exterior algebra: Grade observability, emergent su(2), and fermionic thermodynamics, companion manuscript, in preparation (2026).

\bibitem{Dicke1954CoherenceIS}
R.~H. Dicke, Coherence in spontaneous radiation processes, Phys. Rev. \textbf{93}, 99 (1954).

\bibitem{jafarov2015cpkrawtchouk}
E.~I. Jafarov, A.~M. Jafarova, and J.~Van der Jeugt, The su(2) Krawtchouk oscillator model under the CP deformed symmetry, J. Phys.: Conf. Ser. \textbf{597}, 012047 (2015), arXiv:1502.00464.

\bibitem{genest2013multivariate}
V.~X. Genest, L.~Vinet, and A.~Zhedanov, The multivariate Krawtchouk polynomials as matrix elements of the rotation group representations on oscillator states, arXiv:1306.4256 (2013).

\bibitem{genest2015charlier}
V.~X. Genest, H.~Miki, L.~Vinet, and G.~Yu, A superintegrable discrete harmonic oscillator based on bivariate Charlier polynomials, arXiv:1511.09155 (2015).

\bibitem{mayqin2024algebraic}
M.~Q. May and H.~Qin, Algebraic discrete quantum harmonic oscillator with dynamic resolution scaling, J. Phys. A: Math. Theor. \textbf{57}, 415304 (2024), arXiv:2304.01486.

\bibitem{berger2017koszulsignmap}
R.~Berger, A Koszul sign map, arXiv:1708.01430 (2017).

\end{thebibliography}
\end{document}